\documentclass[twocolumn,showpacs,preprintnumbers,amsmath,amssymb,superscriptaddress]{revtex4}

\usepackage{amsthm}
\usepackage{graphicx}
\usepackage{color}
\usepackage{natbib}
\usepackage{hyperref}
\usepackage{subfigure}
\usepackage{mcode}

\newtheorem{proposition}{Proposition}
\newtheorem{lemma}[proposition]{Lemma}
\newtheorem{theorem}[proposition]{Theorem}
\newtheorem{corollary}[proposition]{Corollary}

\theoremstyle{definition}

\theoremstyle{remark}

\begin{document}
\title{Increasing Peer Pressure on any Connected Graph Leads to Consensus}
\date{\today~-~preprint}
\author{Justin Semonsen}
\email[JS E-mail:]{js2118@math.rutgers.edu}
\affiliation{Department of Mathematics, 
Rutgers University, 
Piscataway, NJ 08854
}
\author{Christopher Griffin}
\email[CG E-mail:]{griffinch@ieee.org}
\affiliation{
Mathematics Department, 
United States Naval Academy, 
Annapolis, MD 21666}
\author{Anna Squicciarini}
\email[AS E-mail:]{asquicciarini@ist.psu.edu}
\affiliation{
College of Info. Sci. and Tech., 
Penn State University, 
University Park, PA 16802}
\author{Sarah Rajtmajer}
\email[SR E-mail: ]{sarah.rajtmajer@qs-2.com}
\affiliation{
Quantitative Scientific Solutions, 
Arlington, VA 22203}

\begin{abstract}

In this paper, we study a model of opinion dynamics in a social network in the presence increasing interpersonal influence, i.e., increasing peer pressure. Each agent in the social network has a distinct social stress function given by a weighted sum of internal and external behavioral pressures. We assume a weighted average update rule and prove conditions under which a connected group of agents converge to a fixed opinion distribution, and under which conditions the group reaches consensus. We show that the update rule is a gradient descent and explain its transient and asymptotic convergence properties. Through simulation, we study the rate of convergence on a scale-free network and then validate the assumption of increasing peer pressure in a simple empirical model.
\end{abstract}

 
\maketitle

\section{Introduction}
Beginning with DeGroot \cite{DeGroot74}, opinion models have been studied extensively (see e.g.,  \cite{FJ90,Krause00,HK02,SL03,BN05,WDA05,Toscani06,Weisb06,Lorenz07,BHT09,CFL09,KR11,DMPW12,CFT12,JM14,Dandekar09042013,Bindel2015248,Bhawalkar:2013:COF:2488608.2488615}). In these models, opinion is a dynamic state variable whose evolution in some compact subset of $\mathbb{R}^n$ is governed by an autonomous dynamical system. Using this formalism, opinion models have been unified with flocking models (see e.g., \cite{TT98,CS07}) in \cite{MT14}. Most recent work on opinion dynamics (and their unification with flocking models) considers the interaction of agents on a graph structure \cite{Dandekar09042013,Bindel2015248,Bhawalkar:2013:COF:2488608.2488615,MT14,JM14}. When considered on a lattice, these models are share characteristics to continuous variations of Ising models \cite{SL03}.

Recent work \cite{Bindel2015248,Bhawalkar:2013:COF:2488608.2488615} considers the evolution of opinion on a social network in which agents are resistant to change because of an innate belief. In particular, \cite{Bindel2015248,Bhawalkar:2013:COF:2488608.2488615} use a variant of the model in \cite{FJ90} and study this problem from a game-theoretic perspective by considering the price of anarchy on the opinion formation process on a connected graph. The existence of innate beliefs, which are hidden but affect (publicly) presented opinion, is supported in recent empirical work by Stephens-Davidowitz et al. \cite{CCHK15,S17,S14}. While the work in \cite{Bindel2015248,Bhawalkar:2013:COF:2488608.2488615} introduces the concept of the stubborn agent, it does not consider the effect of \textit{situationally variant} peer-pressure on agents' opinions, though statically weighted user connections are considered. Peer pressure in social networks is well-documented. Adoption of trends \cite{trends1, trends2}, purchasing behaviors \cite{purchasing}, beliefs and cultural norms \cite{norms1}, privacy behaviors \cite{RSGK16}, bullying \cite{bullying1}, and health behaviors \cite{obesity, weight, sleep, smoking} have all been linked to peer influence. 

In this paper, we 
consider the problem of opinion dynamics on a social network of agents with innate beliefs in which peer-pressure is a dynamically changing quantity, independent of the opinions themselves. This has the mathematical effect of transforming the formerly autonomous dynamical system into a non-autonomous dynamical system. 

 Our notion of persuasion and peer-pressure affecting these dynamics  is related to the psychology literature on belief formation and social influence. In  particular,  we draw inspiration from studies on periodicity in human behavior, and social influence theories \cite{friedkin2006}. We follow Friedkin's foundational theory that strong ties are more likely to affect users' opinions and result in persuasion or social influence. Underpinning our model is also the notion  of mimicking. Brewer and more recently Van Bareen \cite{B91,BJCD09} suggest that mimicking is used when individuals feel out of  a group and therefore will alter their behavior (to a point \cite{BJCD09}) to be more socially accepted.

The resulting   model also accounts for agents with relatively varying resistance to changing their innate beliefs. We use a recent result from functional analysis on the composition of (distinct) contraction mappings along with the Sherman-Morrison formula to show:
\begin{enumerate}
\item Under increasing peer-pressure, the dynamical system converges.
\item If peer-pressure increases in an unbounded way, consensus emerges a weighted average of the innate beliefs of the individuals.
\item The opinion update process converges to a gradient descent, with linear convergence rate
\item The hypothesis of increasing peer-pressure can be supported with a live data set. 
\end{enumerate}
Work herein is complementary to (e.g.) \cite{Bindel2015248,Dandekar09042013,Bhawalkar:2013:COF:2488608.2488615} in that we consider a dynamic (increasing) peer-pressure coefficient with variable weights on initial belief. Additionally, we analyze the convergence rate of the dynamical system to the fixed point, while \cite{Bindel2015248,Bhawalkar:2013:COF:2488608.2488615} focus on the model from a game-theoretic perspective. 

The remainder of this paper is organized as follows: In Section \ref{sec:BasicModel} we present the basic model. In Section \ref{sec:model} we prove convergence of the model and that increasing peer pressure leads to consensus in any connected graph. We discuss the convergence rate in Section \ref{sec:rate} by showing the dynamical system is, effectively, gradient descent. We briefly relate our work to the cost of anarchy work from \cite{Bindel2015248} in Section \ref{sec:anarchy}. In Section \ref{sec:experiment}, we validate the hypothesis of increasing peer pressure by fitting our model to a live data set. Conclusions and future directions are presented in Section \ref{sec:Conclusion}. 

\section{Problem Statement and Model}\label{sec:BasicModel}
We model a network of agents, representing individuals in a social network in which each user communicates with her friends/associates, but not necessarily the entire network. Assume that the agents' network is represented by a simple graph $G = (V,E)$ where vertexes $V$ are agents and edges $E$ are the social connections (communications) between them. It is clear that disconnected sections of the graph are independent, so we assume that $G$ is connected. For the remainder of the paper, let $V = \{1, 2 , \dots n\}$, so $E$ is a subset of the two-element subsets of $V$.
The state of Agent $i$ at time $k$ is a continuous value $x_i^{(k)} \in [0,1]$ that represents disclosed opinion on a bivalent topic (e.g., ``I support gun control'' or ``I like classical music''). Each agent has a constant preference $x^{+}_{i} \in [0,1]$ representing her inherent position on the topic. This may differ from the opinion disclosed to the public. The value $x^+_i$ represents inherent agent bias. Further, Agent $i$ is assigned a non-negative vertex weight $s_{i}$ and positive edge weights $w_{ij}$ respectively for $(i, j) \in E$. The weight $s_i$, termed \textit{stubbornness} \cite{Dandekar09042013}, models the tendency of Agent $i$ to maintain her (private) position $x^+_i$ in public. The edge weights $w_{ij}$ represent friendship affinity.  The set of all disclosed opinions is denoted by the vector $\mathbf{x}^{(k)}$ while the set of constant private preferences is $\mathbf{x}^+$. For the remainder of this paper, we refer to publicly disclosed opinions simply as opinions.

Agent $i$'s state is updated by minimizing its social stress: 
\begin{multline}
J_{i}\left(x_{i}^{(k)},\mathbf{x}^{(k-1)}, k\right) = s_{i}\left(x_{i}^{(k)} - x_{i}^{+}\right)^2\\
+ \rho^{(k)}\sum_{j = 1}^{n} w_{ij}\left(x_{i}^{(k)} - x_{j}^{(k-1)}\right)^2
\label{eqn:SocialStress}
\end{multline}
Here $\rho^{(k)}$ is the peer-pressure coefficient. In the sequel, we assume $\rho^{(k)}$ is an increasing function of $k$. As noted in \cite{Bindel2015248}, under these assumptions, the first order necessary conditions are sufficient for minimizing $J_{i}(x_{i}^{(k)},\mathbf{x}^{(k-1)}, k)$. 
The optimal state for Agent $i$ at time $k$ is then:
\begin{equation}
x_{i}^{(k)} = \frac{s_{i}x_{i}^{+} + \rho^{(k)}\sum_{j = 1}^{n} w_{ij}x_{j}^{(k-1)}}{s_{i} + \rho^{(k)}d_i},
\label{eqn:ExtendedFixedPoint}
\end{equation}
where $d_{i} = \sum_{j = 1}^{n} w_{ij}$ is the weighted degree of vertex $i$. The implied update rule is simply a generalization of the DeGroot model variation found in \cite{FJ90} and generalizes the model in \cite{Bindel2015248} by including the stubbornness coefficient and an increasing peer-pressure term.

Let $\mathbf{A}$ be the $n \times n$ weighted adjacency matrix of $G$. In addition, let $\mathbf{D}$ be the $n \times n$ matrix with $d_{i}$ on the diagonal and let $\mathbf{S}$ be the $n \times n$ matrix with $s_{i}$ on the diagonal. Using these terms, the recurrence in Eq. (\ref{eqn:ExtendedFixedPoint}) can be written as: 
\begin{displaymath}
\mathbf{x}^{(k)} = \left(\mathbf{S} + \rho^{(k)}\mathbf{D}\right)^{-1}\left(\mathbf{S}\mathbf{x}^{+} + \rho^{(k)}\mathbf{A}\mathbf{x}^{(k-1)}\right)
\end{displaymath}

We say that the agents converge to consensus $\bar{\mathbf{x}}$ if there is some $N$ so that for all $n > N$, $\lVert \bar{\mathbf{x}} - \mathbf{x}^{(n)}\rVert < \epsilon$ for some small $\epsilon > 0$. This represents meaningful compromise on the issue under consideration. 

\section{Convergence} \label{sec:model}
In this section, we consider the update rule in Eq. (\ref{eqn:ExtendedFixedPoint}) as a sequence of contraction mappings each with its own fixed point. We then show that all these fixed points converge to a weighted average. The result rests on a variation of the contraction mapping theorem from \cite{G91}. 

\begin{lemma}[See Chapter 13 of \cite{GR01}] If $\mathbf{L} = \mathbf{D}-\mathbf{A}$ is the weighted graph Laplacian, then $\mathbf{L}$ has an eigenvalue $0$ with multiplicity $1$ and a corresponding eigenvector $\mathbf{1}$ where $\mathbf{1}$ is the vector of all $1$'s. \hfill\qed
\end{lemma}

\begin{lemma}\label{invmat}
For any $\rho^{(k)} > 0, \mathbf{S} + \rho^{(k)}\mathbf{L}$ is invertible. 
\end{lemma}

\begin{proof}
By definition, the graph Laplacian is a positive semidefinite symmetric matrix. In addition, the only eigenvector with eigenvalue 0 is the vector of all 1s, written $\mathbf{1}$.

Since $\mathbf{S}$ is symmetric and $s_{i} \geq 0$, $\mathbf{S} + \rho^{(k)}\mathbf{L}$ is positive semidefinite as well. Choose $\mathbf{x} \in \mathbb{R}^{n}$ such that $\mathbf{x}^{T}(\mathbf{S} + \rho^{(k)}\mathbf{L})\mathbf{x} = 0$. Then $\mathbf{x}^{T}(\mathbf{S} + \rho^{(k)}\mathbf{L})\mathbf{x} = \mathbf{x}^{T}\mathbf{S}\mathbf{x} + \rho^{(k)}\mathbf{x}^{T}\mathbf{L}\mathbf{x}$.  Since $\mathbf{S}$ and $\mathbf{L}$ are positive semidefinite and $\rho^{(k)} > 0$, this implies that $\mathbf{x}^{T}\mathbf{S}\mathbf{x} = \mathbf{x}^{T}\mathbf{L}\mathbf{x} = 0$.

Since $\mathbf{L}$ is symmetric, by the spectral theorem it has is an orthonormal basis of eigenvectors $\{{\bf b}_{1}, \dots {\bf b}_{n}\}$ with associated eigenvectors $\{\lambda_{1}, \dots \lambda_{n}\}$. Because $\mathbf{L}$ is positive semidefinite $\lambda_{i} \geq 0$, that $\mathbf{x}^{T}\mathbf{L}\mathbf{x} = \sum_{i = 1}^{n} \lambda_{i} (\mathbf{x}^{T}{\bf b}_{i})^{2}$. And, because $\mathbf{x}^{T}\mathbf{L}\mathbf{x} = 0$, if $\lambda_{i} \neq 0$, $\mathbf{x}^{T}{\bf b}_{i} = 0$.

It follows that $\bf x$ is an eigenvector of $\mathbf{L}$ with eigenvalue 0; that is, $\mathbf{x} = c{\bf 1}$ for some constant $c$, and therefore $\mathbf{x}^{T}\mathbf{S}\mathbf{x} = c^{2}\sum_{i=1}^{n} s_{i}$. Since $s_{i} \geq 0$ and not all $s_{i}$ are zero, we must have $c = 0$, so $\mathbf{x} = {\bf 0}$. Following, $\mathbf{S} + \rho^{(k)}\mathbf{L}$ is positive definite, and therefore invertible.
\end{proof}

Define:
\begin{displaymath}
F_{k}(\mathbf{x}) = \left(\mathbf{S} + \rho^{(k)}\mathbf{D}\right)^{-1}\left(\mathbf{S}\mathbf{x}^{+} + \rho^{(k)}\mathbf{A}\mathbf{x}\right)
\end{displaymath}
and let:
\begin{equation}
G_k = F_k\circ F_{k-1} \circ \cdots \circ F_1
\label{eqn:Gk}
\end{equation}
Then $\mathbf{x}^{(k)} = F_k(\mathbf{\mathbf{x}}^{(k-1)})$ and  $\mathbf{x}^{(k)} = G_k(\mathbf{x}^{(0)})$. That is, iterating these $F_{k}$ captures the evolution of $\mathbf{x}^{(k)}$. We show that for each $k$, $F_k$ is a contraction and therefore has a fixed point by the Banach Fixed Point Theorem \cite{RF09}. We use this result in the proof of Theorem \ref{thm:Convergence}. 

\begin{lemma}\label{contraction}
For all $k, F_{k}$ is a contraction map with fixed point given by $\overline{\mathbf{x}}^{(k)} = (\mathbf{S} + \rho^{(k)}\mathbf{L})^{-1}\mathbf{S}x^{+}$.
\label{lem:Contraction}
\end{lemma}
\begin{proof}
Let $\mathbf{B}$ be the $(n+1) \times (n+1)$ matrix given by adding a row and column to $\rho^{(k)}(\mathbf{S} + \rho^{(k)}\mathbf{D})^{-1}\mathbf{A}$ as follows: 
\begin{displaymath}
\mathbf{B} = \left[\begin{array}{cc} \rho^{(k)}(\mathbf{S} + \rho^{(k)}\mathbf{D})^{-1}\mathbf{A} & (\mathbf{S} + \rho^{(k)}\mathbf{D})^{-1}\mathbf{S}{\bf 1} \\ 0 & 1 \end{array}\right]
\end{displaymath}

The rows of $\mathbf{B}$ sum to $1$. To see this, replace $x^+$ and $x_i^{(k-1)}$ in Eq. (\ref{eqn:ExtendedFixedPoint}) with $1$. Thus $\mathbf{B}$ is a stochastic matrix for a Markov process with a single absorbing state. Since $G$ is connected and not all $s_{i}$ are equal to 0, a transition exists from each state to the steady state; thus from any starting state, convergence to the steady state is guaranteed. This means that $\lim_{i \to \infty}(\rho^{(k)}(\mathbf{S} + \rho^{(k)}\mathbf{D})^{-1}\mathbf{A})^{i} = 0$, so $\mathbf{A}$ is a convergent matrix. Equivalently, if $\lVert \cdot \rVert$ denotes the matrix operator norm, then, $\lVert \rho^{(k)}(\mathbf{S} + \rho^{(k)}\mathbf{D})^{-1}\mathbf{A} \rVert < 1$. Therefore for any $\mathbf{x},\mathbf{y} \in [0,1]^{n}$:
\begin{align*}
\Vert F_{k}(\mathbf{x}) - F_{k}({\bf y}) \Vert 
&= \Vert (\mathbf{S} + \rho^{(k)}\mathbf{D})^{-1}\rho^{(k)}\mathbf{A}(\mathbf{x} - {\bf y}) \Vert\\
&\leq \Vert (\mathbf{S} + \rho^{(k)}\mathbf{D})^{-1}\rho^{(k)}\mathbf{A}\Vert \Vert\mathbf{x} - {\bf y}\Vert
\end{align*}
That is, $F_{k}$ is a contraction map on a compact set, so by the Banach fixed-point theorem, it has a unique fixed point $\overline{\mathbf{x}}^{(k)}$. 

Let $\overline{\mathbf{x}}^{(k)}$ be that fixed point. Then $\overline{\mathbf{x}}^{(k)} = F_{k}(\overline{\mathbf{x}}^{(k)})$. Rearranging the terms yields, 
\begin{displaymath}
(\mathbf{S} + \rho^{(k)}\mathbf{D})\overline{\mathbf{x}}^{(k)} - \rho^{(k)}\mathbf{A}\overline{\mathbf{x}}^{(k)}  = (\mathbf{S} + \rho^{(k)}\mathbf{L}) \overline{\mathbf{x}}^{(k)} = \mathbf{S}\mathbf{x}^{+}.
\end{displaymath}
Therefore:
\begin{equation}
\overline{\mathbf{x}}^{(k)} = (\mathbf{S} + \rho^{(k)}\mathbf{L})^{-1}\mathbf{S}\mathbf{x}^{+}.
\label{eqn:FixedPointk}
\end{equation}
This completes the proof.
\end{proof}

The following lemma will allow us to consider the matrices $(\mathbf{S} + \rho^{(k)}\mathbf{L})^{-1}$ for $k \in \{1,2,\dots\}$ in $\mathrm{GL}_{n}(\mathbb{R})$ (the Lie group of invertible $n\times n$ real matrices) as perturbations.  This enables effective approximations of asymptotic behaviors.
\begin{lemma}\label{shermanmorrison}
Let $\{{\bf b}_{1} \dots {\bf b}_{n}\}$ an orthonormal basis of $\mathbb{R}^{n}$. Also let $\mathbf{M}: \mathbb{R}^{n} \to \mathbb{R}^{n}$ be an invertible symmetric linear transformation (invertible square matrix) and $\{{\bf u}_{1}, \dots {\bf u}_{n}\}$ be a set of unit vectors such that for a small constant $\delta$, $\mathbf{M}^{-1}{\bf b}_{1} = \lambda {\bf b}_{1} + O(\delta){\bf u}_{1}$ and $\mathbf{M}^{-1}{\bf b}_{j} = O(\delta){\bf u}_{j}$ for $j \neq 1$.

Then if $\Vert {\bf v} \Vert = 1$, and $s \in \mathbb{R}$, then unless $(\mathbf{M} + s{\bf v}{\bf v}^{T})$ is not invertible, there exists a set of unit vectors $\{{\bf u}_{1}', \dots {\bf u}_{n}'\}$ such that $(\mathbf{M} + s{\bf v}{\bf v}^{T})^{-1}{\bf b}_{1} = \frac{\lambda}{1 + s\lambda ({\bf v}^{T}{\bf b}_{1})^{2}} {\bf b}_{1} + O(\delta){\bf u}_{1}'$ and $(\mathbf{M} + s{\bf v}{\bf v}^{T})^{-1}{\bf b}_{j} = O(\delta){\bf u}_{j}'$ for $j \neq 1$.
\end{lemma}
Before proceeding to the proof of this result, based on the Sherman-Morrison formula, we note that we will establish an instance of the necessary conditions of this lemma in Theorem \ref{infinitepressure}. Thus the lemma is not vacuous. 
\begin{proof}[Proof of Lemma \ref{shermanmorrison}]
Since $\{{\bf b}_{1} \dots {\bf b}_{n}\}$ is an orthonormal basis, ${\bf v} = \sum_{i=1}^{n} a_{i}{\bf b}_{i}$ where $a_{i} = {\bf v}^{T}{\bf b}_{i}$. This means that $\mathbf{M}^{-1}{\bf v} =  \sum_{i=1}^{n} a_{i}\mathbf{M}^{-1}({\bf b}_{i}) = \lambda a_{1}{\bf b}_{1} + O(\delta)\sum_{i=1}^{n} a_{i}{\bf u}_{i}$. By Cauchy-Schwartz, $\vert a_{i} \vert \leq \Vert {\bf v}\Vert \Vert {\bf b}_{i} \Vert = 1$, so by the triangle inequality, $\Vert \sum_{i=1}^{n} a_{i}{\bf u}_{i} \Vert \leq n$. Then letting ${\bf u} = \frac{\sum_{i=1}^{n} a_{i}{\bf u}_{i}}{n}$, we have that $\mathbf{M}^{-1}{\bf v} = \lambda a_{1}{\bf b}_{1} + O(\delta){\bf u}$, where $\Vert {\bf u} \Vert \leq 1$.

By the Sherman-Morrison formula, 
\begin{displaymath}
(\mathbf{M} + s{\bf v}{\bf v}^{T})^{-1} = \mathbf{M}^{-1} - \frac{s\mathbf{M}^{-1}{\bf v}{\bf v}^{T}\mathbf{M}^{-1}}{1 + s{\bf v}^{T}\mathbf{M}^{-1}{\bf v}}. 
\end{displaymath}
Using this, and choosing each ${\bf u}_{i}'$ to be an appropriate rescaling of the $O(\delta)$ terms yields:

\begin{multline*}
(\mathbf{M} + s{\bf v}{\bf v}^{T})^{-1}({\bf b}_{1}) = \mathbf{M}^{-1}{\bf b}_{1} - \frac{s\mathbf{M}^{-1}{\bf v}{\bf v}^{T}\mathbf{M}^{-1}{\bf b}_{1}}{1 + s{\bf v}^{T}M^{-1}{\bf v}} = \\
\lambda {\bf b}_{1} + O(\delta){\bf u}_{1} - \frac{s{\bf v}^{T}(\lambda {\bf b}_{1} + O(\delta){\bf u}_{1})}{1 + s{\bf v}^{T}(\lambda a_{1}{\bf b}_{1} + O(\delta){\bf u})}(\lambda a_{1}{\bf b}_{1} + O(\delta){\bf u})\\
= \lambda {\bf b}_{1} + O(\delta){\bf u}_{1} - \frac{s\lambda a_{1} + O(\delta)}{1 + s\lambda a_{1}^{2} + O(\delta)}(\lambda a_{1}{\bf b}_{1} + O(\delta){\bf u})\\
= \frac{\lambda}{1 + s\lambda ({\bf v}^{T}{\bf b}_{1})^{2}} {\bf b}_{1} + O(\delta){\bf u}_{1}'
\end{multline*}
Furthermore, for $j \neq 1$:
\begin{multline*}
(\mathbf{M} + s{\bf v}{\bf v}^{T})^{-1}({\bf b}_{j}) = \mathbf{M}^{-1}{\bf b}_{j} - \frac{s\mathbf{M}^{-1}{\bf v}{\bf v}^{T}\mathbf{M}^{-1}{\bf b}_{j}}{1 + s{\bf v}^{T}\mathbf{M}^{-1}{\bf v}} = \\ O(\delta){\bf u}_{j} - \frac{s{\bf v}^{T}(O(\delta){\bf u}_{1})}{1 + s{\bf v}^{T}(\lambda a_{1}{\bf b}_{1} + O(\delta){\bf u})}(\lambda a_{1}{\bf b}_{1} + O(\delta){\bf u})\\
= O(\delta){\bf u}_{j}'
\end{multline*}
This completes the proof.
\end{proof}

The results stated give insight into the motion of fixed points as $\rho^{(k)}$ increases. We now show that the  fixed point given by Eq. (\ref{eqn:FixedPointk}) converge to the average of the agents' initial preferences, weighted by the stubbornness of each agent. We then use that result to prove the dynamics converge to this point when $\rho^{(k)} \rightarrow \infty$. 


\begin{theorem}\label{infinitepressure}
If $\lim_{k \to \infty} \rho^{(k)} = \infty$, then:
\begin{displaymath}
\lim_{k \to \infty}\overline{\mathbf{x}}^{(k)} = \frac{\sum_{i = 1}^{n} s_{i}x_{i}^{+}}{\sum_{i = 1}^{n} s_{i}}{\bf 1}.
 \end{displaymath}
\end{theorem}

\begin{proof}
Since $G$ is a graph, the Laplacian $\mathbf{L}$ is a positive semidefinite symmetric matrix, and therefore has an orthonormal basis of eigenvectors $\{{\bf b}_{1}, \dots {\bf b}_{n}\}$ with real eigenvalues $\{\lambda_{1}, \dots \lambda_{n}\}$. Since $G$ is connected, only a single eigenvalue $\lambda_{1} = 0$ and the associated unit eigenvector is ${\bf b}_{1} = \frac{1}{\sqrt{n}}{\bf 1}$.

Since every vector is an eigenvector of the identity matrix ${\bf I}, \{{\bf b}_{1}, \dots {\bf b}_{n}\}$ are an orthonormal basis of eigenvectors for ${\bf I} + \rho^{(k)}\mathbf{L}$ with eigenvalues $\{1, 1 + \rho^{(k)}\lambda_{2}, \dots 1 + \rho^{(k)}\lambda_{n}\}$. But then $({\bf I} + \rho^{(k)}\mathbf{L})^{-1}$ has the same basis of eigenvectors, with eigenvalues $\{1, \frac{1}{1 + \rho^{(k)}\lambda_{2}}, \dots \frac{1}{1 + \rho^{(k)}\lambda_{n}}\}$.

As $\rho^{(k)} \to \infty, \frac{1}{1 + \rho^{(k)}\lambda_{j}} \to 0$ for each $j \neq 1$. In particular, for any  $\delta > 0$, for sufficiently large $\rho^{(k)}, {\bf I}+ \rho^{(k)}\mathbf{L}$ satisfies the conditions of Lemma \ref{shermanmorrison} with $\lambda = 1$.

Let ${\bf I} + \rho^{(k)}\mathbf{L} = \mathbf{M}_{0}$. Then, for each $l$ up to $n$, let $\mathbf{M}_{l} = (\mathbf{M}_{l-1} + (s_{l} - 1){\bf e}_{l}{\bf e}_{l}^{T})$ where ${\bf e}_{l}$ is the $l$th vector of the standard basis. Since ${\bf e}_{l}{\bf e}_{l}^{T}$ is the zero matrix with a one in the $l$th place on the diagonal, $\sum_{l=1}^{n} (s_{l} - 1){\bf e}_{i}{\bf e}_{i}^{T} = \mathbf{S}-{\bf I}$ and therefore $\mathbf{M}_{n} = ({\bf I} + \rho^{(k)}\mathbf{L}) + \sum_{l=1}^{n} (s_{l} - 1){\bf e}_{l}{\bf e}_{l}^{T} = \mathbf{S} + \rho^{(k)}\mathbf{L}$.

By iterating Lemma \ref{shermanmorrison} with $s = s_{l} - 1$ and ${\bf v} = {\bf e}_{l}$, we have that for each $l$ there is a $\lambda_{l}$ such that $\mathbf{M}_{l}^{-1}{\bf b}_{1} = \lambda_{l} {\bf b}_{1} + O(\delta){\bf u}_{1}^{(l)}$ and $\mathbf{M}_{l}^{-1}{\bf b}_{j} = O(\delta){\bf u}_{j}^{(l)}$ for $j \neq 1$.

Since ${\bf e}_{l}^{T}{\bf b}_{1} = \frac{1}{\sqrt{n}}$, Lemma \ref{shermanmorrison} gives the recurrence: $$\lambda_{l} = \frac{\lambda_{l-1}}{1+ \lambda_{l-1}\frac{s_{i} -1}{n}}$$

Solving this recurrence with $\lambda_{0} = 1$ yields:  
\begin{displaymath}
\mathbf{M}_{l}^{-1}{\bf b}_{1} = \frac{n}{n + \sum_{k=1}^{l} (s_{k} - 1)} {\bf b}_{1} + O(\delta){\bf u}_{1}^{(l)}
\end{displaymath}

Since $\sum_{k=1}^{n} (s_{k} - 1) = tr(\mathbf{S}) - n$, it is clear that
\begin{displaymath}
\mathbf{M}_{n}^{-1}{\bf b}_{1} = \frac{n}{tr(\mathbf{S})} {\bf b}_{1} + O(\delta){\bf u}_{1}^{(n)}.
\end{displaymath}
Therefore, for ${\bf u} = \sum_{i=1}^{n} {\bf b}_{1}^{T}\mathbf{S}\mathbf{x}^{+} {\bf u}_{i}^{(n)}$:
\begin{align*}
\overline{\mathbf{x}}^{(k)} &= \frac{n}{tr(\mathbf{S})} {\bf b}_{1}^{T}\mathbf{S}\mathbf{x}^{+} {\bf b}_{1} + O(\delta) \sum_{i=1}^{n} {\bf b}_{1}^{T}\mathbf{S}\mathbf{x}^{+} {\bf u}_{i}^{(n)}\\
&= \frac{{\bf 1}^{T}\mathbf{S}\mathbf{x}^{+} }{tr(\mathbf{S})} {\bf 1} + O(\delta) {\bf u}\\
&= \frac{\sum_{i=1}^{n} s_{i}x_{i}^{+}}{\sum_{i=1}^{n} s_{i}}{\bf 1} + O(\delta) {\bf u}
\end{align*}

Since $\delta \to 0$ as $\rho^{(k)} \to \infty$, if $\lim_{k \to \infty} \rho^{(k)} = \infty$, then:
\begin{displaymath}
\lim_{k \to \infty} \overline{\mathbf{x}}^{(k)} = \frac{\sum_{i=1}^{n} s_{i}x_{i}^{+}}{\sum_{i=1}^{n} s_{i}}{\bf 1}.
\end{displaymath}
This completes the proof.
\end{proof}

Since peer pressure increases in each step, no single $F_{k}$ is sufficient to model the process of convergence. We use the following result from \cite{G91,Lorentzen90}
\begin{lemma}[Theorem 1 of \cite{Lorentzen90} \& Theorem 2 of \cite{G91}] Let $\{f_n\}$ be a sequence of analytic contractions in a domain $D$ with $f_n(D) \subseteq E \subseteq D_0 \subseteq D$ for all $n$. Then $F_n = f_n \circ f_{n-1} \circ \cdots \circ f_1$ converges uniformly in $D_0$ and locally uniformly in $D$ to a constant function $F(z) = c \in E$. Furthermore, the fixed points of $f_{n}$ converge to the constant $c$. 
\hfill\qed
\label{lem:Contraction2}
\end{lemma}
The following corollary is now immediate from Lemmas \ref{lem:Contraction} and \ref{lem:Contraction2}:
\begin{corollary}\label{convergence}
From Eq. (\ref{eqn:Gk}), let $G_{k}= F_{k} \circ G_{k-1} = F_{k} \circ F_{k-1} \circ \dots \circ F_{1}$ for each $k \geq 0$. Then $G = \lim_{k \to \infty} G_{k}$ is a constant function and (functional) convergence is uniform. 
\label{cor:Convergence}
\end{corollary}
We now have the following theorem, which follows immediately from Corollary \ref{cor:Convergence} and Theorem \ref{infinitepressure}:
\begin{theorem} If $\rho^{(k)} \rightarrow \infty$, then:
\begin{displaymath}
\lim_{k \rightarrow \infty} \mathbf{x}^{(k)} = \frac{\sum_{i = 1}^{n} s_{i}x_{i}^{+}}{\sum_{i = 1}^{n} s_{i}}{\bf 1}.
\end{displaymath}
\hfill\qed
\label{thm:Convergence}
\end{theorem}

This means that in the case of increasing and unbounded peer pressure, all the agents' opinions always converge to consensus. In addition, the value of this consensus is the average of their preferences weighted by their stubbornness. This is irrespective of the weighting of the edges in the network, so long as the network is connected.

We illustrate opinion consensus on a simple graph with 15 vertices in Fig. \ref{fig:SmallGraph}. The vertices are organized into three connected cliques. Each clique was initialized with a distinct range of opinions in $[0,1]$. Initial stubbornness was set randomly and is shown by relative vertex size.
\begin{figure}[htbp]
\centering
\includegraphics[scale=0.45]{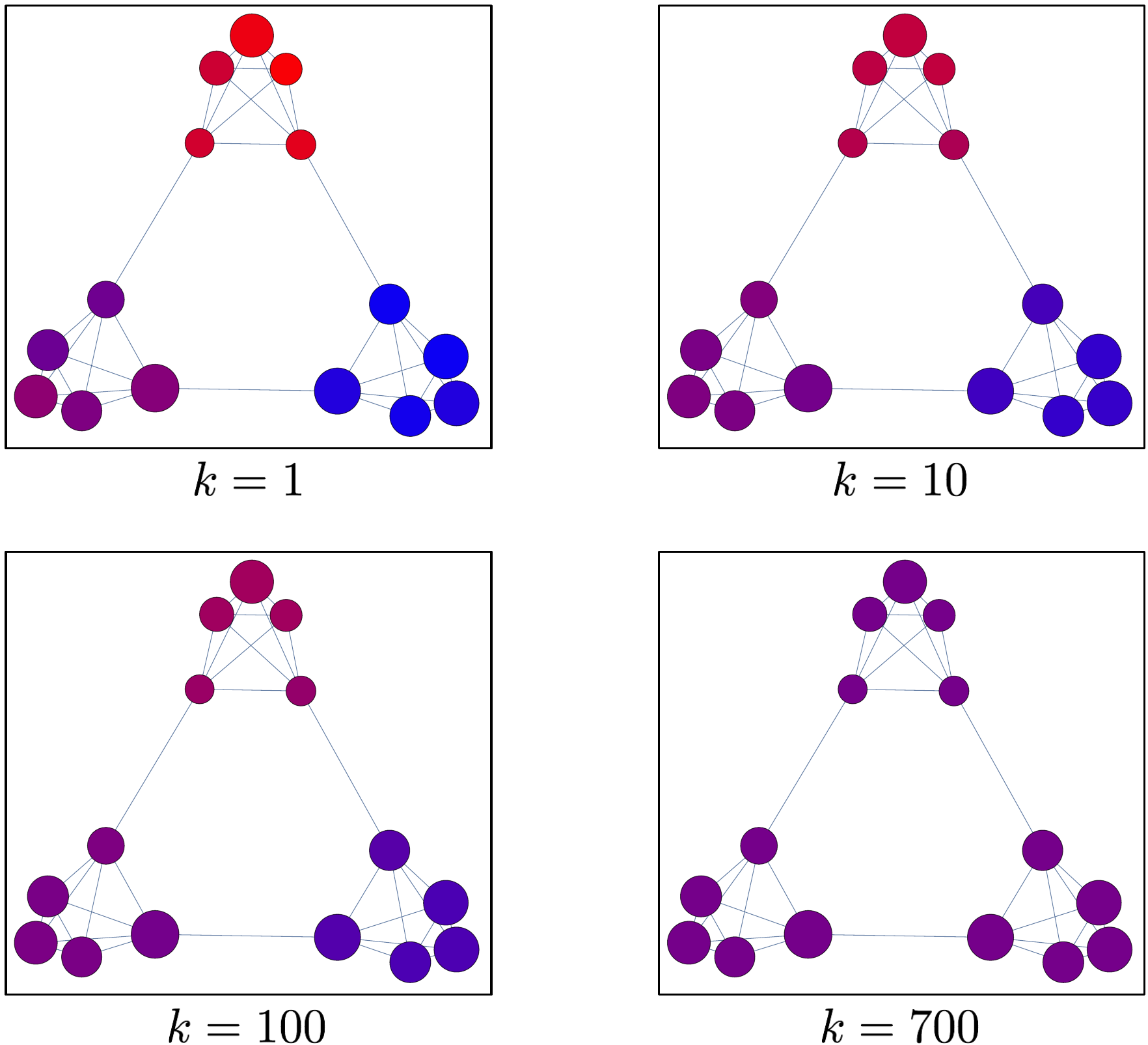}
\caption{The evolution of opinions on a small graph with three cliques. Color indicates opinion, size resistance to change.}
\label{fig:SmallGraph}
\end{figure}
The opinion trajectories for this example are shown in Fig. \ref{fig:SmallOpinions}.
\begin{figure}[htbp]
\centering
\subfigure[Convergence]{\includegraphics[scale=0.25]{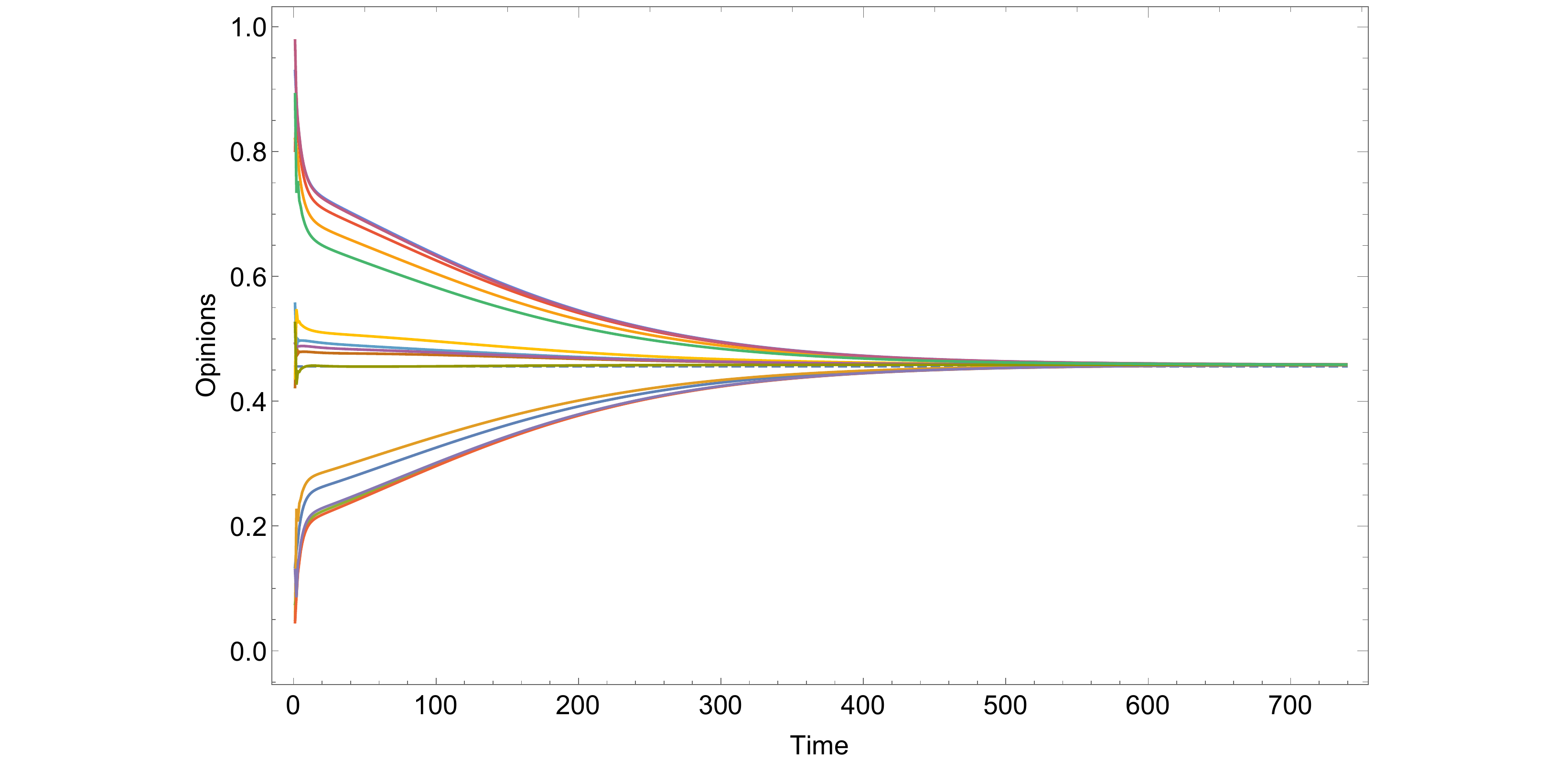}}
\subfigure[Initial Transient Behavior]{\includegraphics[scale=0.26]{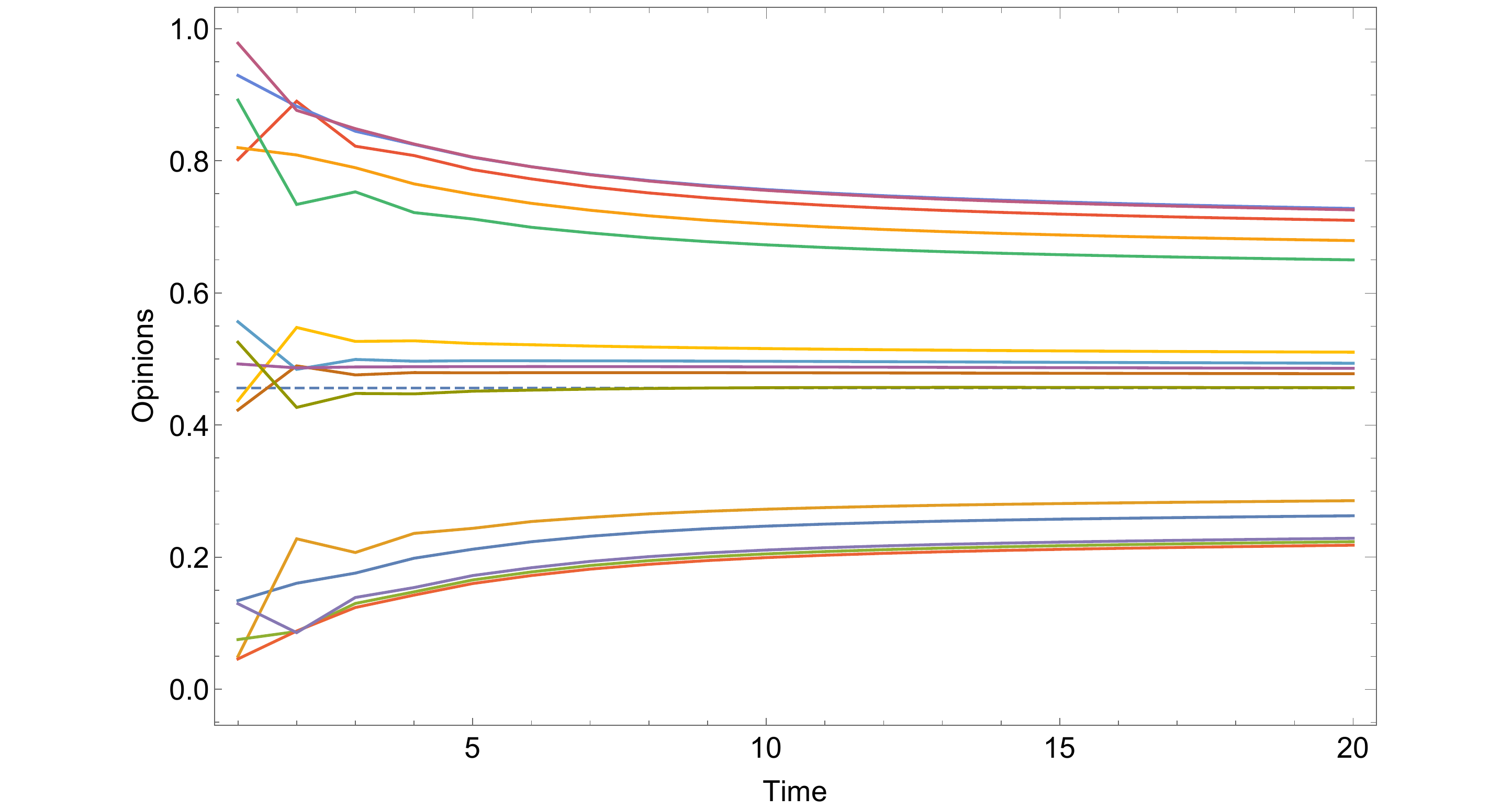}}
\caption{The opinion trajectories are shown illustrating (a) convergences and (b) initial transient behaviors.}
\label{fig:SmallOpinions}
\end{figure}

In the case of increasing but bounded peer pressure, we have:
\begin{displaymath}
\lim_{k \to\infty} \rho^{(k)} \leq \rho^*
\end{displaymath}
Further, this limit always exists by monotone convergence. Intuitively, this means the influence of others is limited, and that personal preferences will always slightly skew the opinions of others. Again, this is consistent with social influence theories on bounded peer pressure and trade-offs with comfort level \cite{BJCD09,Bindel2015248}.

\begin{theorem}\label{finitepressure}
Suppose $\rho^{(k)}$ is increasing and bounded and that:
\begin{displaymath}
\lim_{k \to\infty} \rho^{(k)} = \rho^*,
\end{displaymath}
then 
\begin{displaymath}
\lim_{k \to \infty} \mathbf{x}^{(k)} = (\mathbf{S} + \rho^*\mathbf{L})^{-1}\mathbf{S}\mathbf{x}^{+}.
\end{displaymath}
\end{theorem}
\begin{proof}
Since $\rho^{(k)}$ is increasing and bounded, it converges to a finite number $\rho^*$ by monotone convergence. From Lemma \ref{invmat}, $(\mathbf{S} + \rho^* \mathbf{L})$ is defined and invertible. Since matrix inversion is continuous in $\mathrm{GL}_{n}(\mathbb{R})$, by Theorem \ref{convergence}:
\begin{align*}
\lim_{k \to \infty} \mathbf{x}^{(k)} &= \lim_{k \to \infty} \overline{\mathbf{x}}^{(k)}\\
&= \lim_{k \to \infty}(\mathbf{S} + \rho^{(k)}\mathbf{L})^{-1}\mathbf{S}\mathbf{x}^{+}\\
&= (\mathbf{S} + \lim_{k \to \infty} \rho^{(k)} \mathbf{L})^{-1}\mathbf{S}\mathbf{x}^{+}\\
&= (\mathbf{S} + \rho^* \mathbf{L})^{-1}\mathbf{S}\mathbf{x}^{+}
\end{align*}
\end{proof} 
The above theorem tells us that if peer pressure is increasing and bounded, the agents' opinions converge to a fixed distribution, which may not be a consensus, but is easily computable from the initial preferences. In this case, the shape of the network is important for determining the limit distribution, as the edge weights factor into the Laplacian. This result is similar to the convergence point given in \cite{Bindel2015248} where stubbornness coefficients are not presented and peer pressure is constant.

\section{Convergence Rate}\label{sec:rate}
We analyze the convergence rate of the algorithm and obtain a secondary result on efficiency. Define the utility of these convergent points to be the sum of the stress of the agents when the state $\mathbf{x}$ is constant. Formally:
\begin{equation}
\begin{aligned}
U^{(k)} (\mathbf{x}) &= \sum_{i} J_{i}(x_{i}, \mathbf{x}, k)\\ 
&= \sum_{i =1}^{n} s_{i}(x_{i} - x_{i}^{+})^{2} +  \rho^{k} \left(\sum_{i,j} w_{ij}(x_{i}-x_{j})^{2}\right)\\ 
&= (\mathbf{x} - \mathbf{x}^{+})^{T}\mathbf{S}(\mathbf{x} - \mathbf{x}^{+}) + 2\rho^{k} \mathbf{x}^{T}\mathbf{L}\mathbf{x} \\
&= \mathbf{x}^{T}(\mathbf{S} + 2\rho^{k}\mathbf{L})\mathbf{x} - 2\mathbf{x}^{T}\mathbf{S} \mathbf{x}^{+} + (\mathbf{x}^{+})^{T}\mathbf{S} \mathbf{x}^{+}
\end{aligned}
\label{eqn:GlobalUtility}
\end{equation}
Define the limiting utility $U(\mathbf{x})$ as:
\begin{equation}
U(\mathbf{x}) = \lim_{k \rightarrow \infty} \frac{1}{\rho^{(k)}}U^{(k)}(\mathbf{x})
\end{equation}
The following lemma is immediately clear from the construction of the functions $J_i$, the fact that $U^{(k)}$ is a strictly convex function and $U$ is the limit of these strictly convex functions:
\begin{lemma} The global utility function $U(\mathbf{x})$ is convex. Furthermore, the fact that (i) $U^{(k)}$ is smooth on its entire domain and (ii) $U^{(k)}(\mathbf{x})/\rho^{(k)}$ converges uniformly to $U(\mathbf{x})$, implies that $U(\mathbf{x})$ is both differentiable and its derivative can be computed as the limit of the derivatives of 
$U^{(k)}(\mathbf{x})/\rho^{(k)}$.
\hfill\qed
\end{lemma}

Using the global utility function, we can analyze the convergence rate of the update rule. From Eq. (\ref{eqn:ExtendedFixedPoint}), we can compute:
\begin{multline}
\Delta x_{i}^{(k-1)} = x_i^{(k)} - x_i^{(k-1)} = \\
\frac{s_i(x_i^+ - x_i^{(k-1)}) + \sum_{j=1}^n \left(x_j^{(k-1)} -x_i^{(k-1)}\right)}{s_i + \rho^{(k)}\sum_{j=1}^n w_{ij}}
\end{multline}
Let:
\begin{equation}
\alpha_i^{(k)} = \frac{1}{s_i + \rho^{(k)}\sum_{j=1}^n w_{ij}}
\end{equation}
and define $\mathbf{H}^{(k)} = \tfrac{1}{2}\mathrm{diag}\left(\alpha_1^{(k)},\dots,\alpha_n^{(k)}\right)$. Computing the gradient of $U^{(k)}$ yields:
\begin{equation}
\Delta \mathbf{x} = -\mathbf{H}^{(k)}\nabla U^{(k)}\left(\mathbf{x}_i^{(k-1)}\right)
\end{equation}
We conclude the update rule, Eq. (\ref{eqn:ExtendedFixedPoint}) can be written:
\begin{equation}
\mathbf{x}^{(k)} = \mathbf{x}^{(k-1)} - \mathbf{H}^{(k)}\nabla U^{(k)}\left(\mathbf{x}^{k-1}\right)
\label{eqn:GradientUpdate}
\end{equation}
Necessarily, $\mathbf{H}^{(k)}$ is always positive definite and therefore $- \mathbf{H}^{(k)}\nabla U^{(k)}\left(\mathbf{x}^{k-1}\right)$ is always a descent direction for $U^{(k)}$. Moreover, $(\nabla U_k)^T \nabla U > 0$ and consequently $- \mathbf{H}^{(k)}\nabla U^{(k)}\left(\mathbf{x}^{k-1}\right)$ is a descent direction for $U(\mathbf{x})$. Thus, the update rule is a descent algorithm, which explains the initial fast convergence toward the average (see Fig. \ref{fig:SmallOpinions}).  When the descent direction converges to a Newton step, a descent algorithm can be shown to converge superlinearly  \cite{Bert99}. However, these steps do not converge to Newton steps. As $\rho^{(k)}$ grows large, $\alpha_i^{(k)} \rightarrow 0$ and $U_k/\rho^{(k)} \rightarrow U$ and consequently for large $k$:
\begin{displaymath}
\frac{1}{\rho^{(k)}}\mathbf{H}^{(k)}\nabla U^{(k)}\left(\mathbf{x}^{k-1}\right) \approx \epsilon \nabla U\left(\mathbf{x}^{k-1}\right)
\end{displaymath}
for $\epsilon \sim 1/\rho^{(k)}$. Thus, the update rule approaches a simple gradient descent. We show that a consequence of this is a linear convergence rate. 

Let:
\begin{displaymath}
\mathbf{x}^*= \frac{\sum_{i=1}^{n} s_{i}x_{i}^{+}}{\sum_{i=1}^{n} s_{i}}{\bf 1}
\end{displaymath}
and define:
\begin{displaymath}
\mathbf{y}^{(k)} = \mathbf{x}^{(k)} - \mathbf{x}^*.
\end{displaymath}
From Eq. (\ref{eqn:GradientUpdate}) we compute:
\begin{equation}
\frac{\left\lVert\mathbf{x}^{(k+1)} - \mathbf{x}^*\right\rVert}{\left\lVert\mathbf{x}^{(k)} - \mathbf{x}^*\right\rVert} = 
\frac{\left\lVert\mathbf{y}^{(k)} - \mathbf{H}^{(k+1)}\nabla U^{(k+1)}(\mathbf{x}^{(k)})\right\rVert}{\left\lVert\mathbf{y}^{(k-1)} - \mathbf{H}^{(k)}\nabla U^{(k)}(\mathbf{x}^{(k-1)})\right\rVert}
\end{equation}
Assuming $\rho^{(k)} \rightarrow \infty$ as $k \rightarrow \infty$, and expanding the gradient using Eq. (\ref{eqn:GlobalUtility}) we obtain:
\begin{widetext}
\begin{multline}
\lim_{k \rightarrow \infty}\frac{\left\lVert\mathbf{x}^{(k+1)} - \mathbf{x}^*\right\rVert}{\left\lVert\mathbf{x}^{(k)} - \mathbf{x}^*\right\rVert}  = \lim_{k \rightarrow \infty} \frac{\tfrac{1}{\rho^{(k)}}}{\tfrac{1}{\rho^{(k)}}}\frac{\left\lVert\mathbf{y}^{(k)} - \mathbf{H}^{(k+1)}\left(\left[\mathbf{S} + 2\rho^{(k+1)}\mathbf{L}\right]\mathbf{x}^{(k)} - 2\mathbf{S}\mathbf{x}^{+}\right)\right\rVert}{\left\lVert\mathbf{y}^{(k-1)} - \mathbf{H}^{(k)}\left(\left[\mathbf{S} + 2\rho^{(k)}\mathbf{L}\right]\mathbf{x}^{(k-1)} - 2\mathbf{S}\mathbf{x}^{+}\right)\right\rVert} =\\ 
\lim_{k \rightarrow \infty}
\frac{\left\lVert\mathbf{y}^{(k)}/\rho^{(k)} - \mathbf{H}^{(k+1)}\left(\left[\mathbf{S}/\rho^{(k)} + 2\tfrac{\rho^{(k+1)}}{\rho^{(k)}}\mathbf{L}\right]\mathbf{x}^{(k)} - 2\mathbf{S}\mathbf{x}^{+}/\rho^{(k)}\right)\right\rVert}{\left\lVert\mathbf{y}^{(k-1)}/\rho^{(k)} - \mathbf{H}^{(k)}\left(\left[\mathbf{S}/\rho^{(k)} + 2\mathbf{L}\right]\mathbf{x}^{(k-1)} - 2\mathbf{S}\mathbf{x}^{+}/\rho^{(k)}\right)\right\rVert} = \\
\lim_{k \rightarrow \infty}\frac{2\tfrac{\rho^{(k+1)}}{\rho^{(k)}}\left\lVert \mathbf{H}^{(k+1)}\mathbf{L}\mathbf{x}^{(k)}\right\rVert}
{2\left\lVert\mathbf{H}^{(k)}\mathbf{L}\mathbf{x}^{(k-1)}\right\rVert}
\end{multline}
\end{widetext}
As $\rho^{(k)} \rightarrow \infty$, we see that:
\begin{displaymath}
\mathbf{H}^{(k)} \rightarrow \frac{1}{2\rho^{(k)}}\mathbf{D}^{-1},
\end{displaymath}
where $\mathbf{D}$ is the diagonal weighted degree matrix. Then:
\begin{multline*}
\lim_{k \rightarrow \infty} \frac{2\tfrac{\rho^{(k+1)}}{\rho^{(k)}}\left\lVert \mathbf{H}^{(k+1)}\mathbf{L}\mathbf{x}^{(k)}\right\rVert}
{2\left\lVert\mathbf{H}^{(k)}\mathbf{L}\mathbf{x}^{(k-1)}\right\rVert} = \\  \frac{\rho^{(k+1)}}{\rho^{(k)}}\frac{\tfrac{1}{\rho^{(k+1)}}\left\lVert\mathbf{D}^{-1}\mathbf{L}\mathbf{x}^{(k)}\right\rVert}{\tfrac{1}{\rho^{(k)}}\left\lVert\mathbf{D}^{-1}\mathbf{L}\mathbf{x}^{(k-1)}\right\rVert} = 1
\end{multline*}
Thus we have shown:
\begin{theorem} The convergence rate of the update rule given in Eq. (\ref{eqn:ExtendedFixedPoint}) is linear. In particular: 
\begin{equation}
\lim_{k \rightarrow \infty}\frac{\left\lVert\mathbf{x}^{(k+1)} - \mathbf{x}^*\right\rVert}{\left\lVert\mathbf{x}^{(k)} - \mathbf{x}^*\right\rVert} = 1,
\label{eqn:RateOfConvergence}
\end{equation}
\end{theorem}

We illustrate the slow convergence on a larger example with 500 vertices organized into a scale-free graph using the Barab\'{a}si-Albert \cite{BA99} graph construction algorithm. The graph and snapshots of opinion evolution are shown in Fig. \ref{fig:BigGraphs}.
\begin{figure}[htbp]
\centering
\includegraphics[scale=0.4]{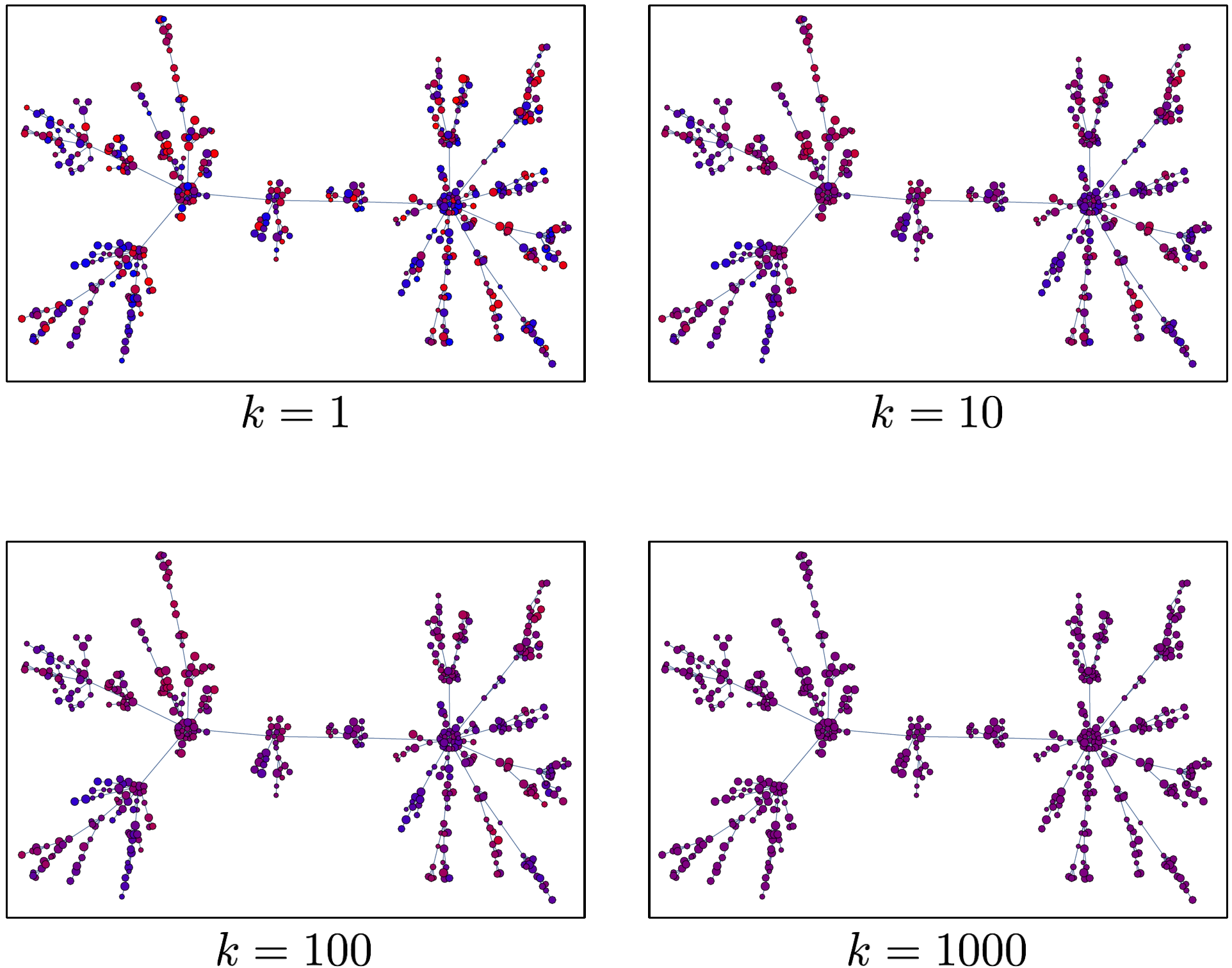}
\caption{A larger example of convergence using a scale-free graph with 500 vertices generated using the Barab\'{a}si-Albert algorithm.}.
\label{fig:BigGraphs}
\end{figure}
We show the opinion trajectories for the 500 vertex scale-free network in Fig. \ref{fig:BigGraphConvergence}(a) and illustrate Eq. (\ref{eqn:RateOfConvergence}) in Fig. \ref{fig:BigGraphConvergence}(b).
\begin{figure}[htbp]
\centering
\subfigure[Convergence]{\includegraphics[scale=0.36]{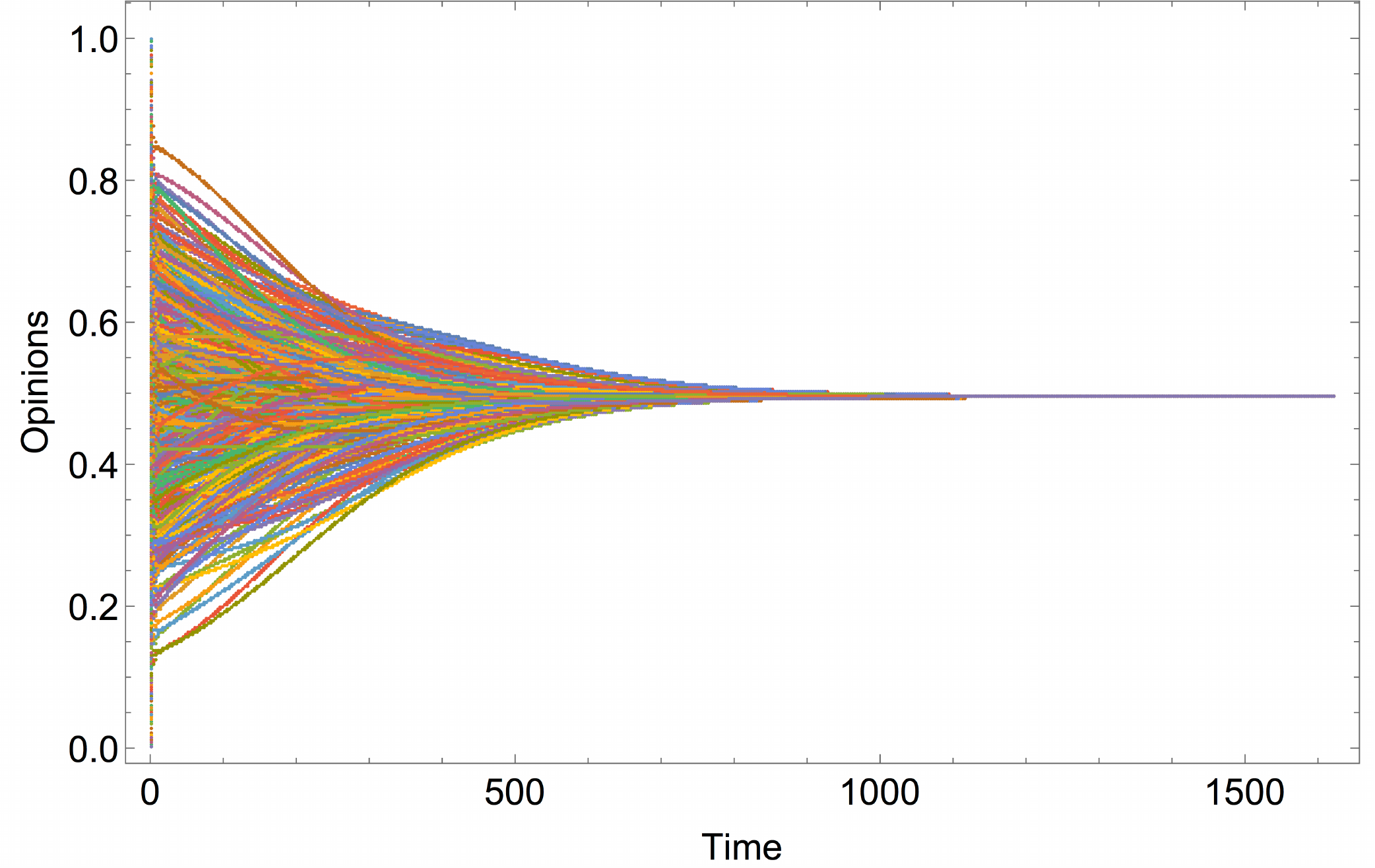}}
\subfigure[Convergence Ratio]{\includegraphics[scale=0.3]{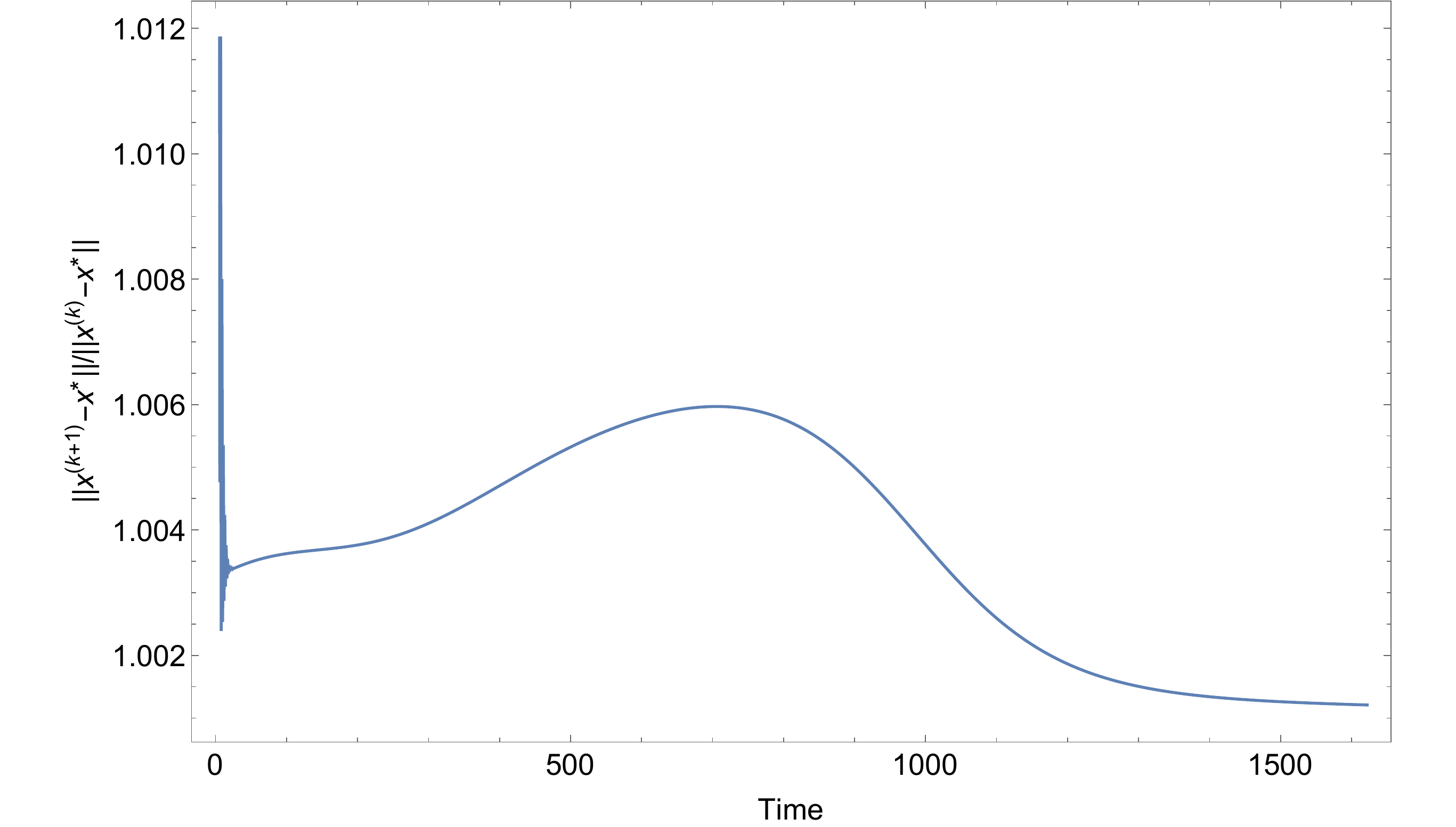}}
\caption{(a)Opinion convergence is illustrated in a scale free network. (b) Linear convergence is demonstrated by showing the ratio ${\left\lVert\mathbf{x}^{(k+1)} - \mathbf{x}^*\right\rVert}/{\left\lVert\mathbf{x}^{(k)} - \mathbf{x}^*\right\rVert}$ converges to 1 as expected.}
\label{fig:BigGraphConvergence}
\end{figure}
Notice the ratio ${\left\lVert\mathbf{x}^{(k+1)} - \mathbf{x}^*\right\rVert}/{\left\lVert\mathbf{x}^{(k)} - \mathbf{x}^*\right\rVert}$ approaches $1$ as expected.

\section{Cost of Anarchy}\label{sec:anarchy}
\cite{Bindel2015248} observe that simultaneous minimization of Eq. (\ref{eqn:SocialStress}) is a game-theoretic problem and compare the total social utility in a centralized solution to a decentralized solution (Nash equilibrium); i.e., they compute a price of anarchy \cite{Rough03,Bindel2015248}. To analyze the price of anarchy of this system, we cannot use the utility function in Eq. (\ref{eqn:GlobalUtility}), as the $U^{(k)}(\lim_{k \to \infty} \mathbf{x}^{(k)}) \to 0$ when $\rho^{(k)} \to \infty$. Instead, we use a total utility function $U_T(\mathbf{x}) = \lim_{k\to \infty} U^{(k)}(\mathbf{x})$ to compute the cost of anarchy:
\begin{theorem}\label{optimality}
 The convergent point $\lim_{k \to \infty} \mathbf{x}^{(k)}$ minimizes total utility if and only if $\lim_{k \to \infty} \rho^{k} = \infty$.
\end{theorem}
\begin{proof}
If $\rho^{(k)}$ converges to a finite number $\rho^*$, then the total utility is
$$U_T(x) = \mathbf{x}^{T}(\mathbf{S} + 2\rho^*\mathbf{L})\mathbf{x} - 2\mathbf{x}^{T}\mathbf{S} \mathbf{x}^{+} + (\mathbf{x}^{+})^{T}\mathbf{S} \mathbf{x}^{+}$$
Note that this is identical to the work in \cite{Bindel2015248}, except with edge weights multiplied by $\rho^*$. We note that $\lim_{k \to \infty} \mathbf{x}^{(k)}$ is the Nash Equilibrium used in \cite{Bindel2015248}. From the work in \cite{Bindel2015248} we may conclude the convergent point is not optimal for finite $\rho^*$.

If $\lim_{k \to \infty} \rho^{k} = \infty$, then if $\mathbf{x} \neq c\mathbf{1}$ for some constant $c$, then $\mathbf{x}^{T}\mathbf{L}\mathbf{x} > 0$, so $U^{(k)}(\mathbf{x})$ grows without bound. However, for any $k$, we have that $U^{(k)}(c\mathbf{1}) = \sum_{i=1}^n s_i(c - x_i^+)^2$, so $U_T(c\mathbf{1}) = \sum_{i=1}^n s_i(c - x_i^+)^2$. By first order necessary conditions of optimality:
 $$\frac{\sum_{i=1}^{n} s_{i}x_{i}^{+}}{\sum_{i=1}^{n} s_{i}}{\mathbf 1}$$ 
 minimizes $U(\mathbf{x})$, and thus $\lim_{k \to \infty} \mathbf{x}^{(k)}$ is optimal.
\end{proof}
This gives the following trivial corollary, which is consistent with the work in \cite{Bindel2015248}.
\begin{corollary} The cost of anarchy is $1$ if and only if $\lim_{k \to \infty} \rho^{(k)} = \infty$. \hfill\qed
\end{corollary}

\section{Empirical Analysis} \label{sec:experiment}
The hypothesis of increasing peer pressure in social settings underlies this work. We attempt to (in)validate the hypothesis that peer-pressure does increase in real-world systems, by using data from the well-known \textit{Social Evolution Experiment} \cite{madan2012}. The experiment tracked the everyday life of approximately 80 students in an undergraduate dormitory over 6 months using mobile phones and surveys, in order to mine spatio-temporal behavioral patterns and the co-evolution of individual behaviors and social network structure.

The dataset includes proximity, location, and call logs, collected through a mobile application. Also included are sociometric survey data for relationships,  political opinions, recent smoking behavior, attitudes towards exercise and fitness, attitudes towards diet, attitudes towards academic performance, current confidence and anxiety level, and musical tastes. 

The derived social network graph (shown Fig. \ref{fig:SocialNet}) represents each student as a node; an edge is present between two nodes if either student noted any level of interaction during the surveys. Edge weights were derived based on the level of interaction recorded between the students in the surveys, as well as the number of surveys in which the interaction appeared. 
\begin{figure}[htbp]
\centering
\subfigure[Graph Layout]{\includegraphics[scale=0.3]{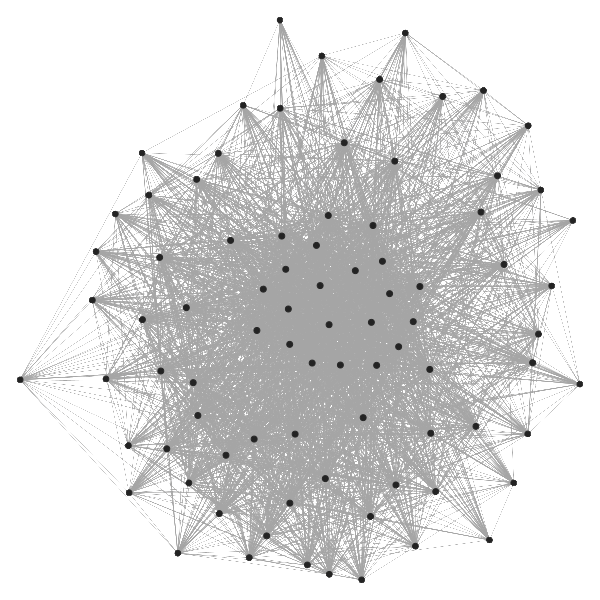}}
\subfigure[Degree Distribution]{\includegraphics[scale=0.35]{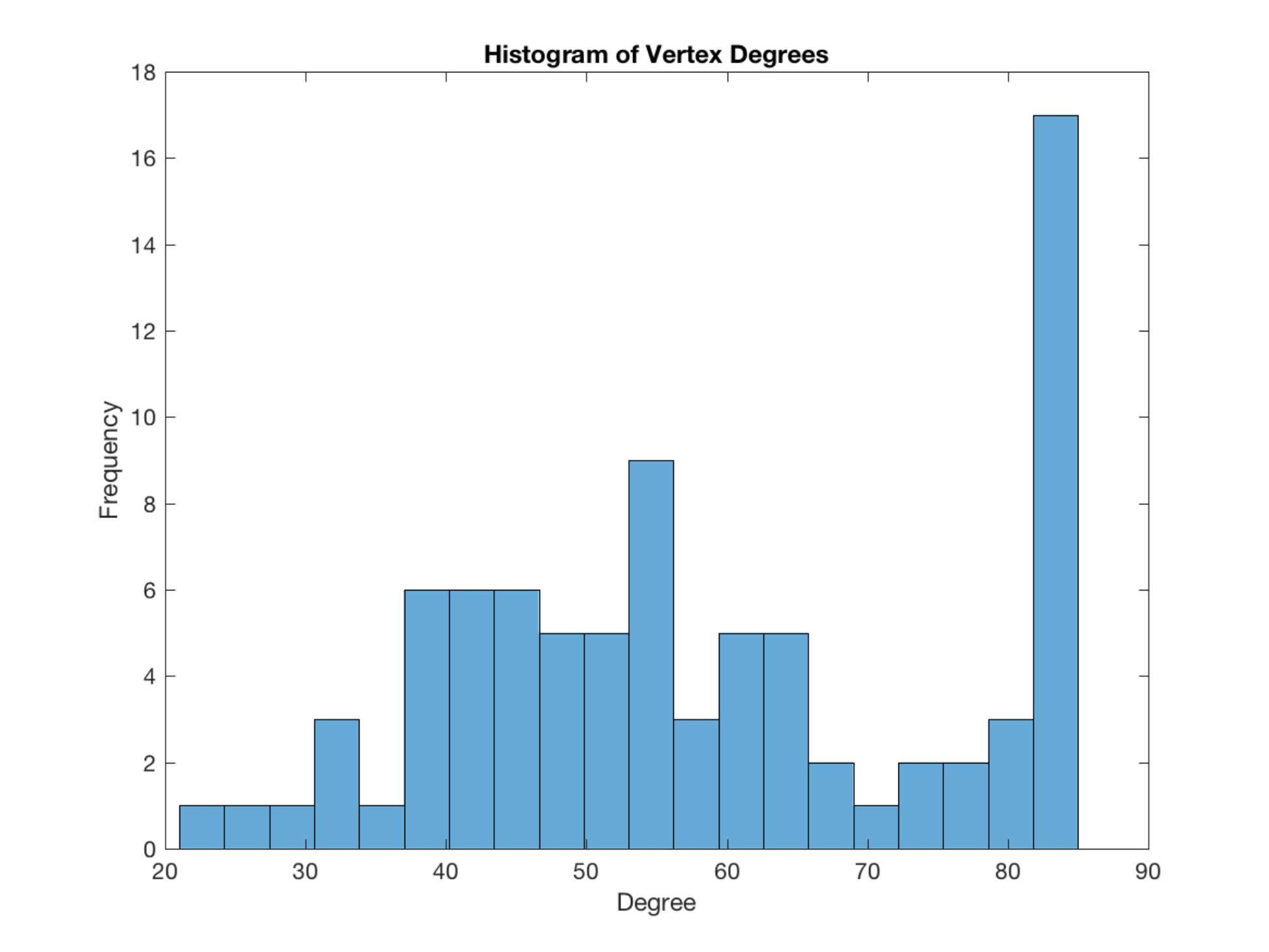}}
\caption{Social network of the political data set showing a high-degree of social connectedness.}
\label{fig:SocialNet}
\end{figure}
We note that this graph is not scale-free, as is typical of social networks. This may be a result of the size of the network, collection bias or simply representative of this social network. As a consequence, it is dense.

Political opinion was modeled on a $[0,1]$ scale, with lower numbers representing Republican preferences and higher numbers representing Democratic preferences. Individual scores were assigned based on reported political party, preferred candidate and likelihood of voting (prior to the election), as well as who they voted for and their approval rating of Barack Obama (after the election). Appendix \ref{sec:AppendixA} contains the code used to set these preferences. Each month's survey was examined individually to put together a monthly time line of each person's political views. The results of the first survey were used as proxy for their inherent personal preference, prior to peer influence.

Finally, individual stubbornness/lack of susceptibility to peer pressure was approximated using reported interest in politics on the first survey administered, as well as stated likelihood of voting. These survey questions were independent of those used in determining political preferences. Appendix \ref{sec:AppendixB} contains the code used to set stubbornness.

Given a list of $\rho^{(k)}$, students' preferences were simulated by aligning each iteration of play to one day in the survey period. The simulated preferences were compared to the surveyed preferences at each month, and the distances between the vectors were summed to get a single score for each list of $\rho^{(k)}$. 

This function of $\rho^{(k)}$ was minimized with fminsearch in Matlab and the best-fit peer pressure values were found to be increasing, with a best fit line $\rho^{(k)} = 1.06*k - 11.96$ and $r^{2} = .9886$ (see Fig. \ref{fig:FitLines}).
\begin{figure}[htbp]
\centering
\includegraphics[scale=0.3]{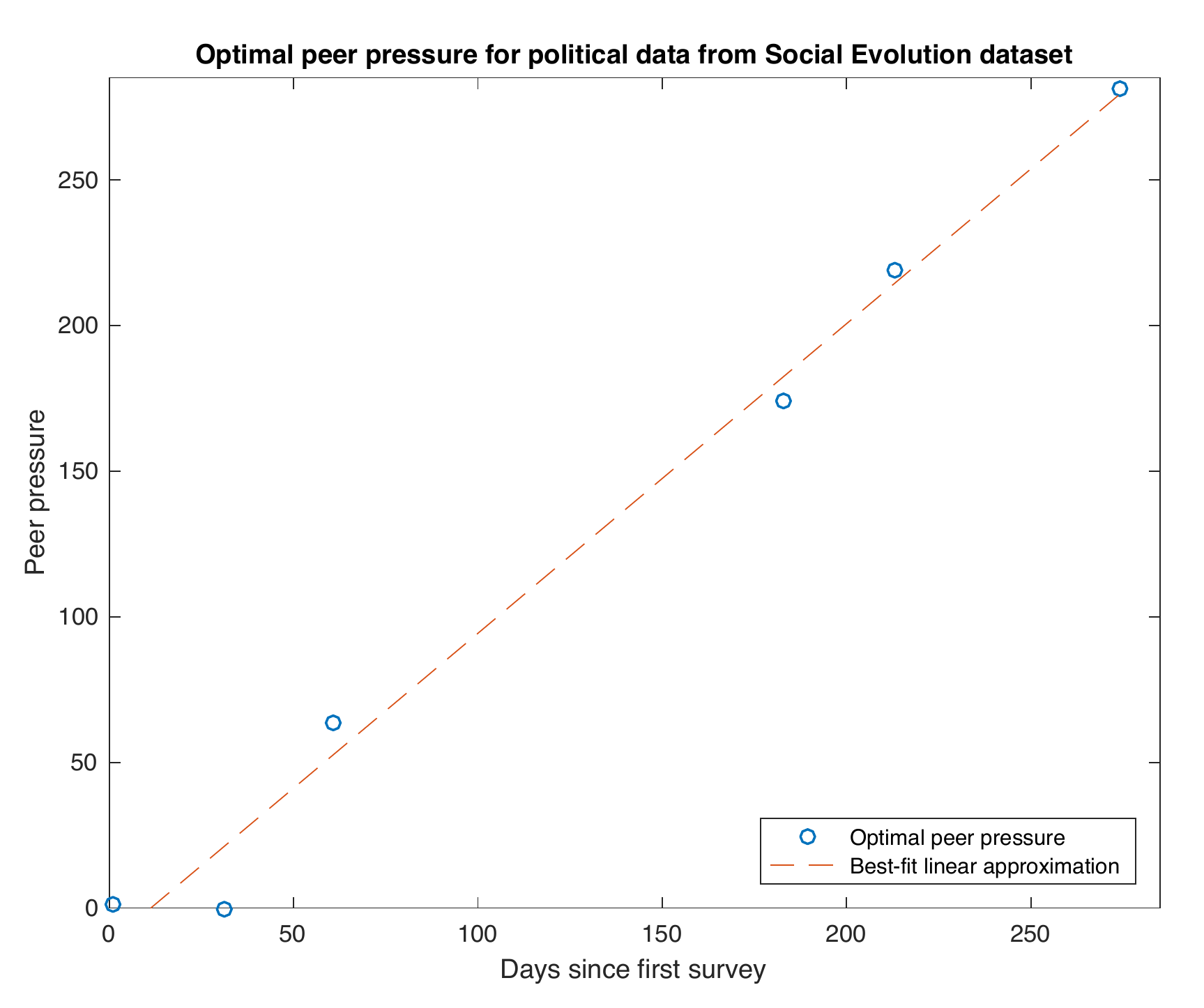}
\caption{Fit of estimated peer pressure showing a clear increase over time, validating the primary hypothesis of the paper that individual consensus occurs because peer pressure increases on each round.}
\label{fig:FitLines}
\end{figure}
The inferred increasing in peer pressure is consistent with the underlying hypothesis of the paper, under the assumption that the process of repeated opinion averaging with stubbornness is a valid model of human behavior. We discuss this further in Section \ref{sec:Conclusion}.

\section{Conclusion}\label{sec:Conclusion}
In this paper we study an opinion formation model under the presence of increasing peer pressure. As in earlier work, we consider agents whose opinion is affected by unchanging innate beliefs. In this paper, the relative strength of these innate beliefs may vary from agent to agent. We show that in the case of unbounded peer-pressure, opinion consensus to a weighted average of innate beliefs is ensured. We also consider the case when peer-pressure is increasing, but bounded. Simulation suggests a numerically slow convergence, which is explained by showing the system dynamics converge to gradient descent applied to a certain convex function. Using this observation we show that that convergence is linear. We evaluate our hypothesis that peer-pressure increases in real world closed systems by fitting our model to a live data set. 

We note that that the assumption of a non-constant (and increasing) peer-pressure coefficient can help mitigate the fast initial convergence of this class of models. It is rare in the real-world to see dramatic opinion shifts over extremely short time scales. Such dramatic shifts are consistent with a gradient descent. However, by varying peer-pressure, the gradient descent can be controlled, leading to more consistency with real-world phenomena as illustrated.

In future work, the limitation that the network is undirected and symmetric should be removed to account for asymmetric social influence. In addition, the network is assumed to remain static during the convergence process, with connections independent of the agents' opinions. Sufficiently different opinions could cause enough stress between agents so as to cause them to reduce influence or even sever the tie between them. A dynamic network model as in \cite{CGT16} could accommodate this kind of network update. Finally, it would be interesting to study corresponding control problems, in which we are given a $\bar{\mathbf{x}}$, the desired convergence point and we can control a subset of agents reporting values ($x_i^{(k)}$), stubbornness ($s_i$) or initial value ($x_i^+$) to determine conditions under which opinion steering is possible. This problem becomes more interesting if the other agents attempt to determine whether certain agents are intentionally attempting to manipulate the opinion value $\mathbf{x}^{(k)}$. Of equal interest is the transient control problem in which $\mathbf{x}^{(k)}$ is steered through a set $X \subset \mathbb{R}^n$ under the assumption that external factors will prevent convergence in the long-run.

\section*{Acknowledgement}
All authors were supported in part by the Army Research Office, under Grant W911NF-13-1-0271. A portion of CG's work was supported by the National Science Foundation under grant number CMMI-1463482. A portion of AS's work was supported by the National Science foundation under grant number 1453080.

\appendix
\section{Initial Condition Code}\label{sec:AppendixA}
The Matlab code below sets the initial preferences ($\mathbf{x}^+$) in this experiment.
\begin{lstlisting}
%Get the months in which the surveys were taken
dates = unique(PoliticalTable.survey_month);

%Set each person's intitial political leanings to exactly center
%This is what we assume if no other information is collected.
InferredPreferences = ones(NumActors,size(dates,1))/2;

for j = 1:size(dates,1)

  %import the political part of the surveys
  I = zeros(size(PoliticalTable,1),1);

  for i = 1:size(PoliticalTable,1)
    I(i) = strcmp(PoliticalTable.survey_month(i),dates(j));
  end

  TempTable = PoliticalTable(I==1,:);

  % For each person, if they say they are liberal
  % or conservative, add  or subtract .1
  % an additional .05 if they say 'extremely'
  % and .05 less if they say 'slightly'
  for i = 1:size(TempTable,1)
    if strcmp(TempTable.liberal_or_conservative(i),...
        'Extremely liberal')
      InferredPreferences(TempTable.user_id(i),j) = .65;
    elseif strcmp(TempTable.liberal_or_conservative(i),...
        'Liberal')
      InferredPreferences(TempTable.user_id(i),j) = .6;
    elseif strcmp(TempTable.liberal_or_conservative(i),...
        'Slightly liberal')
      InferredPreferences(TempTable.user_id(i),j) = .55;
    elseif strcmp(TempTable.liberal_or_conservative(i),...
        'Extremely conservative')
      InferredPreferences(TempTable.user_id(i),j) = .35;
    elseif strcmp(TempTable.liberal_or_conservative(i),...
        'Conservative')
      InferredPreferences(TempTable.user_id(i),j) = .4;
    elseif strcmp(TempTable.liberal_or_conservative(i),...
        'Slightly conservative')
      InferredPreferences(TempTable.user_id(i),j) = .45;
    end

    %Based on their party membership, add or subtract .1
    if strcmp(TempTable.preferred_party(i), 'Democrat')
      InferredPreferences(TempTable.user_id(i),j) =...
        InferredPreferences(TempTable.user_id(i),j) + .1;
    elseif strcmp(TempTable.preferred_party(i), 'Republican')
      InferredPreferences(TempTable.user_id(i),j) =...
        InferredPreferences(TempTable.user_id(i),j) - .1;
    end

    % Add or subtract .1 if they add their party in
    % the details part, with .05 adjustment for strong
    % or not strong. Indpendent and
    %others were not accounted for
    if strcmp(TempTable.preferred_party_details(i),...
        'Strong Democrat')
      InferredPreferences(TempTable.user_id(i),j) =...
        InferredPreferences(TempTable.user_id(i),j) + .15;
    elseif strcmp(TempTable.preferred_party_details(i),...
        'Democratic Party')
      InferredPreferences(TempTable.user_id(i),j) =...
        InferredPreferences(TempTable.user_id(i),j) + .1;
    elseif strcmp(TempTable.preferred_party_details(i),...
        'Not very strong Democrat')
      InferredPreferences(TempTable.user_id(i),j) =...
        InferredPreferences(TempTable.user_id(i),j) + .05;
    elseif strcmp(TempTable.preferred_party_details(i),...
        'Not very strong Republican')
      InferredPreferences(TempTable.user_id(i),j) =...
        InferredPreferences(TempTable.user_id(i),j) - .05;
    elseif strcmp(TempTable.preferred_party_details(i),...
        'Republican Party')
      InferredPreferences(TempTable.user_id(i),j) =...
        InferredPreferences(TempTable.user_id(i),j) - .1;
    elseif strcmp(TempTable.preferred_party_details(i),...
        'Strong Republican')
      InferredPreferences(TempTable.user_id(i),j) =...
        InferredPreferences(TempTable.user_id(i),j) - .15;
    end

    % Add or subtract .2 if the person voted
    % (or would have voted) for one of the major
    % party candidates
    if strcmp(TempTable.did_vote_in_election(i),...
        'I voted in this presidential election')
      if strcmp(TempTable.voted_preferred_candidate(i),...
          'John McCain')
        InferredPreferences(TempTable.user_id(i),j) =...
          InferredPreferences(TempTable.user_id(i),j) - .2;
      elseif strcmp(TempTable.voted_preferred_candidate(i),...
          'Barack Obama')
        InferredPreferences(TempTable.user_id(i),j) =...
          InferredPreferences(TempTable.user_id(i),j) + .2;
      end
    else
      if strcmp(TempTable.not_voted_preferred_candidate(i),...
          'John McCain')
        InferredPreferences(TempTable.user_id(i),j) =...
          InferredPreferences(TempTable.user_id(i),j) - .2;
      elseif strcmp(...
        TempTable.not_voted_preferred_candidate(i),...
          'Barack Obama')
        InferredPreferences(TempTable.user_id(i),j) =...
          InferredPreferences(TempTable.user_id(i),j) + .2;
      end
    end

    %When voting in the election, add or
    % subtract .05 for the party that they
    % vote for, with and extra .1 if they
    % are sure
    if strcmp(TempTable.voting_for_today(i),...
        'Definitely Barack Obama')
      InferredPreferences(TempTable.user_id(i),j) =...
        InferredPreferences(TempTable.user_id(i),j) + .15;
    elseif strcmp(TempTable.voting_for_today(i),...
        'Probably Barack Obama')
      InferredPreferences(TempTable.user_id(i),j) =...
        InferredPreferences(TempTable.user_id(i),j) + .05;
    elseif strcmp(TempTable.voting_for_today(i),...
        'Probably John McCain')
      InferredPreferences(TempTable.user_id(i),j) =...
        InferredPreferences(TempTable.user_id(i),j) - .05;
    elseif strcmp(TempTable.voting_for_today(i),...
        'Definitely John McCain')
      InferredPreferences(TempTable.user_id(i),j) =...
        InferredPreferences(TempTable.user_id(i),j) - .15;
    end

    %After the election add or subtract .3
    % if they approve or diapprove of Obama as
    % president, with a .1 adjustment for
    % strong or slight
    if strcmp(TempTable.approve_obama_president(i),...
        'Strongly Approve')
      InferredPreferences(TempTable.user_id(i),j) = ...
        InferredPreferences(TempTable.user_id(i),j) + .4;
    elseif strcmp(TempTable.approve_obama_president(i),...
        'Slightly Approve')
      InferredPreferences(TempTable.user_id(i),j) =...
        InferredPreferences(TempTable.user_id(i),j) + .2;
    elseif strcmp(TempTable.approve_obama_president(i),...
        'Slightly Disapprove')
      InferredPreferences(TempTable.user_id(i),j) =...
        InferredPreferences(TempTable.user_id(i),j) - .2;
    elseif strcmp(TempTable.approve_obama_president(i),...
        'Strongly Disapprove')
      InferredPreferences(TempTable.user_id(i),j) =...
        InferredPreferences(TempTable.user_id(i),j) - .4;
    end
  end
end

% Finally, restrict the preferences to within
% [0,1] by taking any higher and setting to 1,
% and any lower and setting to 0
InferredPreferences(InferredPreferences > 1) = 1;
InferredPreferences(InferredPreferences < 0) = 0;

%Restrict to usable surveys
InferredPreferences = InferredPreferences(:,[1:3,5:7]);

%And set initial preferences based on these inferences.
InitialPreferences = InferredPreferences(:,1);
\end{lstlisting}

\section{Stubbornness Setting Code}\label{sec:AppendixB}
The Matlab code below set the stubbornness coefficients ($\mathbf{s}$) in this experiment.
\begin{lstlisting}
%Set all initial stubbornness values to 0
s = zeros(NumActors, 1);

% Get the relvant survey data
I = zeros(size(PoliticalTable,1),1);

for i = 1:size(PoliticalTable,1)
  I(i) = strcmp(PoliticalTable.survey_month(i),...
    '2008.09');
end

InitialTable = PoliticalTable(I==1,:);

%For each user
for i = 1:size(InitialTable,1)
  % set their stubbornness values higher if
  % they extress greater interest in politics.
  if strcmp(InitialTable.interested_in_politics(i),...
      'Very interested')
    s(InitialTable.user_id(i)) = 6;
  elseif strcmp(InitialTable.interested_in_politics(i),...
      'Somewhat interested')
    s(InitialTable.user_id(i)) = 4;
  elseif strcmp(InitialTable.interested_in_politics(i),...
      'Slightly interested')
    s(InitialTable.user_id(i)) = 2;
  end

  %Adjust if they say they are more likely
  % to vote (set to 0 if they wont.
  if strcmp(InitialTable.likelihood_of_voting(i),...
      'I will definitely vote')
    s(InitialTable.user_id(i)) =...
      s(InitialTable.user_id(i)) + 3;
  elseif strcmp(InitialTable.likelihood_of_voting(i),...
      'I will most likely vote')
    s(InitialTable.user_id(i)) =...
      s(InitialTable.user_id(i)) + 1;
  elseif strcmp(InitialTable.likelihood_of_voting(i),...
      'I will most likely not vote') &&...
      s(InitialTable.user_id(i))
    s(InitialTable.user_id(i)) =...
      s(InitialTable.user_id(i)) - 1;
  elseif strcmp(InitialTable.likelihood_of_voting(i),...
      'I will definitely not vote')
    s(InitialTable.user_id(i)) = 0;
  end

  % and a little higher if they express interest
  % in a party, regardless of the party.
  if strcmp(InitialTable.preferred_party_details(i),...
      'Strong Democrat')
    s(InitialTable.user_id(i)) =...
      s(InitialTable.user_id(i)) + .15;
  elseif strcmp(InitialTable.preferred_party_details(i),...
      'Democratic Party')
    s(InitialTable.user_id(i)) =...
      s(InitialTable.user_id(i)) + .05;
  elseif strcmp(InitialTable.preferred_party_details(i),...
      'Republican Party')
    s(InitialTable.user_id(i)) =...
      s(InitialTable.user_id(i)) + .05;
  elseif strcmp(InitialTable.preferred_party_details(i),...
      'Strong Republican')
    s(InitialTable.user_id(i)) =...
      s(InitialTable.user_id(i)) + .15;
  end
end

%Initialize the matrix
s = s*StubbornnessConstant;
S = diag(s);
\end{lstlisting}

\bibliographystyle{apsrev4-1}
\bibliography{mrabbrev,bibliography}

\begin{thebibliography}{47}%
\makeatletter
\providecommand \@ifxundefined [1]{%
 \@ifx{#1\undefined}
}%
\providecommand \@ifnum [1]{%
 \ifnum #1\expandafter \@firstoftwo
 \else \expandafter \@secondoftwo
 \fi
}%
\providecommand \@ifx [1]{%
 \ifx #1\expandafter \@firstoftwo
 \else \expandafter \@secondoftwo
 \fi
}%
\providecommand \natexlab [1]{#1}%
\providecommand \enquote  [1]{``#1''}%
\providecommand \bibnamefont  [1]{#1}%
\providecommand \bibfnamefont [1]{#1}%
\providecommand \citenamefont [1]{#1}%
\providecommand \href@noop [0]{\@secondoftwo}%
\providecommand \href [0]{\begingroup \@sanitize@url \@href}%
\providecommand \@href[1]{\@@startlink{#1}\@@href}%
\providecommand \@@href[1]{\endgroup#1\@@endlink}%
\providecommand \@sanitize@url [0]{\catcode `\\12\catcode `\$12\catcode
  `\&12\catcode `\#12\catcode `\^12\catcode `\_12\catcode `\%12\relax}%
\providecommand \@@startlink[1]{}%
\providecommand \@@endlink[0]{}%
\providecommand \url  [0]{\begingroup\@sanitize@url \@url }%
\providecommand \@url [1]{\endgroup\@href {#1}{\urlprefix }}%
\providecommand \urlprefix  [0]{URL }%
\providecommand \Eprint [0]{\href }%
\providecommand \doibase [0]{http://dx.doi.org/}%
\providecommand \selectlanguage [0]{\@gobble}%
\providecommand \bibinfo  [0]{\@secondoftwo}%
\providecommand \bibfield  [0]{\@secondoftwo}%
\providecommand \translation [1]{[#1]}%
\providecommand \BibitemOpen [0]{}%
\providecommand \bibitemStop [0]{}%
\providecommand \bibitemNoStop [0]{.\EOS\space}%
\providecommand \EOS [0]{\spacefactor3000\relax}%
\providecommand \BibitemShut  [1]{\csname bibitem#1\endcsname}%
\let\auto@bib@innerbib\@empty
\bibitem [{\citenamefont {DeGroot}(1974)}]{DeGroot74}%
  \BibitemOpen
  \bibfield  {author} {\bibinfo {author} {\bibfnamefont {M.~H.}\ \bibnamefont
  {DeGroot}},\ }\href@noop {} {\bibfield  {journal} {\bibinfo  {journal} {J.
  American Stat. Association}\ }\textbf {\bibinfo {volume} {69}},\ \bibinfo
  {pages} {118} (\bibinfo {year} {1974})}\BibitemShut {NoStop}%
\bibitem [{\citenamefont {Friedkin}\ and\ \citenamefont
  {Johnsen}(1990)}]{FJ90}%
  \BibitemOpen
  \bibfield  {author} {\bibinfo {author} {\bibfnamefont {N.~E.}\ \bibnamefont
  {Friedkin}}\ and\ \bibinfo {author} {\bibfnamefont {E.~C.}\ \bibnamefont
  {Johnsen}},\ }\href {\doibase 10.1080/0022250X.1990.9990069} {\bibfield
  {journal} {\bibinfo  {journal} {The Journal of Mathematical Sociology}\
  }\textbf {\bibinfo {volume} {15}},\ \bibinfo {pages} {193} (\bibinfo {year}
  {1990})},\ \Eprint
  {http://arxiv.org/abs/http://dx.doi.org/10.1080/0022250X.1990.9990069}
  {http://dx.doi.org/10.1080/0022250X.1990.9990069} \BibitemShut {NoStop}%
\bibitem [{\citenamefont {Krause}(2000)}]{Krause00}%
  \BibitemOpen
  \bibfield  {author} {\bibinfo {author} {\bibfnamefont {U.}~\bibnamefont
  {Krause}},\ }in\ \href@noop {} {\emph {\bibinfo {booktitle} {In
  Communications in Difference Equations}}},\ \bibinfo {editor} {edited by\
  \bibinfo {editor} {\bibnamefont {Gordon}}\ and\ \bibinfo {editor}
  {\bibnamefont {Breach}}}\ (\bibinfo {year} {2000})\ pp.\ \bibinfo {pages}
  {227-- 236}\BibitemShut {NoStop}%
\bibitem [{\citenamefont {Hegselmann}\ and\ \citenamefont
  {Krause}(2002)}]{HK02}%
  \BibitemOpen
  \bibfield  {author} {\bibinfo {author} {\bibfnamefont {R.}~\bibnamefont
  {Hegselmann}}\ and\ \bibinfo {author} {\bibfnamefont {U.}~\bibnamefont
  {Krause}},\ }\href@noop {} {\bibfield  {journal} {\bibinfo  {journal} {J.
  Artificial Soc. Social Simul.}\ }\textbf {\bibinfo {volume} {5}} (\bibinfo
  {year} {2002})}\BibitemShut {NoStop}%
\bibitem [{\citenamefont {Slanina}\ and\ \citenamefont {Lavicka}(2003)}]{SL03}%
  \BibitemOpen
  \bibfield  {author} {\bibinfo {author} {\bibfnamefont {F.}~\bibnamefont
  {Slanina}}\ and\ \bibinfo {author} {\bibfnamefont {H.}~\bibnamefont
  {Lavicka}},\ }\href {\doibase 10.1140/epjb/e2003-00278-0} {\bibfield
  {journal} {\bibinfo  {journal} {The European Physical Journal B - Condensed
  Matter and Complex Systems}\ }\textbf {\bibinfo {volume} {35}},\ \bibinfo
  {pages} {279} (\bibinfo {year} {2003})}\BibitemShut {NoStop}%
\bibitem [{\citenamefont {Ben-Naim}(2005)}]{BN05}%
  \BibitemOpen
  \bibfield  {author} {\bibinfo {author} {\bibfnamefont {E.}~\bibnamefont
  {Ben-Naim}},\ }\href@noop {} {\bibfield  {journal} {\bibinfo  {journal}
  {Europhys. Lett.}\ }\textbf {\bibinfo {volume} {69}},\ \bibinfo {pages} {671}
  (\bibinfo {year} {2005})}\BibitemShut {NoStop}%
\bibitem [{\citenamefont {Weisbuch}\ \emph {et~al.}(2005)\citenamefont
  {Weisbuch}, \citenamefont {Deffuant},\ and\ \citenamefont {Amblard}}]{WDA05}%
  \BibitemOpen
  \bibfield  {author} {\bibinfo {author} {\bibfnamefont {G.}~\bibnamefont
  {Weisbuch}}, \bibinfo {author} {\bibfnamefont {G.}~\bibnamefont {Deffuant}},
  \ and\ \bibinfo {author} {\bibfnamefont {F.}~\bibnamefont {Amblard}},\
  }\href@noop {} {\bibfield  {journal} {\bibinfo  {journal} {Physica A}\
  }\textbf {\bibinfo {volume} {353}} (\bibinfo {year} {2005})}\BibitemShut
  {NoStop}%
\bibitem [{\citenamefont {Toscani}(2006)}]{Toscani06}%
  \BibitemOpen
  \bibfield  {author} {\bibinfo {author} {\bibfnamefont {G.}~\bibnamefont
  {Toscani}},\ }\href@noop {} {\bibfield  {journal} {\bibinfo  {journal}
  {Commun. Math. Sci.}\ }\textbf {\bibinfo {volume} {4}},\ \bibinfo {pages}
  {481} (\bibinfo {year} {2006})}\BibitemShut {NoStop}%
\bibitem [{\citenamefont {Weisbuch}(2006)}]{Weisb06}%
  \BibitemOpen
  \bibfield  {author} {\bibinfo {author} {\bibfnamefont {G.}~\bibnamefont
  {Weisbuch}},\ }in\ \href@noop {} {\emph {\bibinfo {booktitle} {Econophysics
  and Sociophysics: Trends and Perspectives}}},\ \bibinfo {editor} {edited by\
  \bibinfo {editor} {\bibfnamefont {B.~K.}\ \bibnamefont {Chakrabarti}},
  \bibinfo {editor} {\bibfnamefont {A.}~\bibnamefont {Chakrabarti}}, \ and\
  \bibinfo {editor} {\bibfnamefont {A.}~\bibnamefont {Chatterjee}}}\ (\bibinfo
  {publisher} {Wiley},\ \bibinfo {year} {2006})\ pp.\ \bibinfo {pages}
  {67--94}\BibitemShut {NoStop}%
\bibitem [{\citenamefont {Lorenz}(2007)}]{Lorenz07}%
  \BibitemOpen
  \bibfield  {author} {\bibinfo {author} {\bibfnamefont {J.}~\bibnamefont
  {Lorenz}},\ }\href@noop {} {\bibfield  {journal} {\bibinfo  {journal}
  {Internat. J. Modern Phys. C}\ }\textbf {\bibinfo {volume} {18}},\ \bibinfo
  {pages} {1819} (\bibinfo {year} {2007})}\BibitemShut {NoStop}%
\bibitem [{\citenamefont {Blondel}\ \emph {et~al.}(2009)\citenamefont
  {Blondel}, \citenamefont {Hendrickx},\ and\ \citenamefont
  {Tsitsiklis}}]{BHT09}%
  \BibitemOpen
  \bibfield  {author} {\bibinfo {author} {\bibfnamefont {V.~D.}\ \bibnamefont
  {Blondel}}, \bibinfo {author} {\bibfnamefont {J.~M.}\ \bibnamefont
  {Hendrickx}}, \ and\ \bibinfo {author} {\bibfnamefont {J.~N.}\ \bibnamefont
  {Tsitsiklis}},\ }\href {\doibase 10.1109/TAC.2009.2031211} {\bibfield
  {journal} {\bibinfo  {journal} {IEEE Transactions on Automatic Control}\
  }\textbf {\bibinfo {volume} {54}},\ \bibinfo {pages} {2586} (\bibinfo {year}
  {2009})}\BibitemShut {NoStop}%
\bibitem [{\citenamefont {Castellano}\ \emph {et~al.}(2009)\citenamefont
  {Castellano}, \citenamefont {Fortunato},\ and\ \citenamefont
  {Loreto}}]{CFL09}%
  \BibitemOpen
  \bibfield  {author} {\bibinfo {author} {\bibfnamefont {C.}~\bibnamefont
  {Castellano}}, \bibinfo {author} {\bibfnamefont {S.}~\bibnamefont
  {Fortunato}}, \ and\ \bibinfo {author} {\bibfnamefont {V.}~\bibnamefont
  {Loreto}},\ }\href@noop {} {\bibfield  {journal} {\bibinfo  {journal} {Rev.
  Modern Phys.}\ }\textbf {\bibinfo {volume} {81}},\ \bibinfo {pages} {591}
  (\bibinfo {year} {2009})}\BibitemShut {NoStop}%
\bibitem [{\citenamefont {Kurz}\ and\ \citenamefont {Rambau}(2011)}]{KR11}%
  \BibitemOpen
  \bibfield  {author} {\bibinfo {author} {\bibfnamefont {S.}~\bibnamefont
  {Kurz}}\ and\ \bibinfo {author} {\bibfnamefont {J.}~\bibnamefont {Rambau}},\
  }\href@noop {} {\bibfield  {journal} {\bibinfo  {journal} {J. Difference Equ.
  Appl.}\ }\textbf {\bibinfo {volume} {17}},\ \bibinfo {pages} {859} (\bibinfo
  {year} {2011})}\BibitemShut {NoStop}%
\bibitem [{\citenamefont {Duering}\ \emph {et~al.}(2012)\citenamefont
  {Duering}, \citenamefont {Markowich}, \citenamefont {Pietschmann},\ and\
  \citenamefont {Wolfram}}]{DMPW12}%
  \BibitemOpen
  \bibfield  {author} {\bibinfo {author} {\bibfnamefont {B.}~\bibnamefont
  {Duering}}, \bibinfo {author} {\bibfnamefont {P.}~\bibnamefont {Markowich}},
  \bibinfo {author} {\bibfnamefont {J.~F.}\ \bibnamefont {Pietschmann}}, \ and\
  \bibinfo {author} {\bibfnamefont {M.~T.}\ \bibnamefont {Wolfram}},\
  }\href@noop {} {\bibfield  {journal} {\bibinfo  {journal} {Proc. R. Soc.
  Lond. Ser. A}\ }\textbf {\bibinfo {volume} {465}} (\bibinfo {year}
  {2012})}\BibitemShut {NoStop}%
\bibitem [{\citenamefont {Canuto}\ \emph {et~al.}(2012)\citenamefont {Canuto},
  \citenamefont {Fagnani},\ and\ \citenamefont {Tilli}}]{CFT12}%
  \BibitemOpen
  \bibfield  {author} {\bibinfo {author} {\bibfnamefont {C.}~\bibnamefont
  {Canuto}}, \bibinfo {author} {\bibfnamefont {F.}~\bibnamefont {Fagnani}}, \
  and\ \bibinfo {author} {\bibfnamefont {P.}~\bibnamefont {Tilli}},\
  }\href@noop {} {\bibfield  {journal} {\bibinfo  {journal} {SIAM J. Contr. and
  Opt.}\ ,\ \bibinfo {pages} {243}} (\bibinfo {year} {2012})}\BibitemShut
  {NoStop}%
\bibitem [{\citenamefont {Jabin}\ and\ \citenamefont {Motsch}(2014)}]{JM14}%
  \BibitemOpen
  \bibfield  {author} {\bibinfo {author} {\bibfnamefont {P.-E.}\ \bibnamefont
  {Jabin}}\ and\ \bibinfo {author} {\bibfnamefont {S.}~\bibnamefont {Motsch}},\
  }\href {\doibase https://doi.org/10.1016/j.jde.2014.08.005} {\bibfield
  {journal} {\bibinfo  {journal} {Journal of Differential Equations}\ }\textbf
  {\bibinfo {volume} {257}},\ \bibinfo {pages} {4165 } (\bibinfo {year}
  {2014})}\BibitemShut {NoStop}%
\bibitem [{\citenamefont {Dandekar}\ \emph {et~al.}(2013)\citenamefont
  {Dandekar}, \citenamefont {Goel},\ and\ \citenamefont
  {Lee}}]{Dandekar09042013}%
  \BibitemOpen
  \bibfield  {author} {\bibinfo {author} {\bibfnamefont {P.}~\bibnamefont
  {Dandekar}}, \bibinfo {author} {\bibfnamefont {A.}~\bibnamefont {Goel}}, \
  and\ \bibinfo {author} {\bibfnamefont {D.~T.}\ \bibnamefont {Lee}},\ }\href
  {\doibase 10.1073/pnas.1217220110} {\bibfield  {journal} {\bibinfo  {journal}
  {Proceedings of the National Academy of Sciences}\ }\textbf {\bibinfo
  {volume} {110}},\ \bibinfo {pages} {5791} (\bibinfo {year} {2013})},\ \Eprint
  {http://arxiv.org/abs/http://www.pnas.org/content/110/15/5791.full.pdf}
  {http://www.pnas.org/content/110/15/5791.full.pdf} \BibitemShut {NoStop}%
\bibitem [{\citenamefont {Bindel}\ \emph {et~al.}(2015)\citenamefont {Bindel},
  \citenamefont {Kleinberg},\ and\ \citenamefont {Oren}}]{Bindel2015248}%
  \BibitemOpen
  \bibfield  {author} {\bibinfo {author} {\bibfnamefont {D.}~\bibnamefont
  {Bindel}}, \bibinfo {author} {\bibfnamefont {J.}~\bibnamefont {Kleinberg}}, \
  and\ \bibinfo {author} {\bibfnamefont {S.}~\bibnamefont {Oren}},\ }\href
  {\doibase https://doi.org/10.1016/j.geb.2014.06.004} {\bibfield  {journal}
  {\bibinfo  {journal} {Games and Economic Behavior}\ }\textbf {\bibinfo
  {volume} {92}},\ \bibinfo {pages} {248 } (\bibinfo {year}
  {2015})}\BibitemShut {NoStop}%
\bibitem [{\citenamefont {Bhawalkar}\ \emph {et~al.}(2013)\citenamefont
  {Bhawalkar}, \citenamefont {Gollapudi},\ and\ \citenamefont
  {Munagala}}]{Bhawalkar:2013:COF:2488608.2488615}%
  \BibitemOpen
  \bibfield  {author} {\bibinfo {author} {\bibfnamefont {K.}~\bibnamefont
  {Bhawalkar}}, \bibinfo {author} {\bibfnamefont {S.}~\bibnamefont
  {Gollapudi}}, \ and\ \bibinfo {author} {\bibfnamefont {K.}~\bibnamefont
  {Munagala}},\ }in\ \href {\doibase 10.1145/2488608.2488615} {\emph {\bibinfo
  {booktitle} {Proceedings of the Forty-fifth Annual ACM Symposium on Theory of
  Computing}}},\ \bibinfo {series and number} {STOC '13}\ (\bibinfo
  {publisher} {ACM},\ \bibinfo {address} {New York, NY, USA},\ \bibinfo {year}
  {2013})\ pp.\ \bibinfo {pages} {41--50}\BibitemShut {NoStop}%
\bibitem [{\citenamefont {Toner}\ and\ \citenamefont {Tu}(1998)}]{TT98}%
  \BibitemOpen
  \bibfield  {author} {\bibinfo {author} {\bibfnamefont {J.}~\bibnamefont
  {Toner}}\ and\ \bibinfo {author} {\bibfnamefont {Y.}~\bibnamefont {Tu}},\
  }\href@noop {} {\bibfield  {journal} {\bibinfo  {journal} {Physical Review
  E}\ }\textbf {\bibinfo {volume} {58}},\ \bibinfo {pages} {4828} (\bibinfo
  {year} {1998})}\BibitemShut {NoStop}%
\bibitem [{\citenamefont {Cucker}\ and\ \citenamefont {Smale}(2007)}]{CS07}%
  \BibitemOpen
  \bibfield  {author} {\bibinfo {author} {\bibfnamefont {F.}~\bibnamefont
  {Cucker}}\ and\ \bibinfo {author} {\bibfnamefont {S.}~\bibnamefont {Smale}},\
  }\href@noop {} {\bibfield  {journal} {\bibinfo  {journal} {IEEE Transactions
  on Automatic Control}\ }\textbf {\bibinfo {volume} {52}},\ \bibinfo {pages}
  {852} (\bibinfo {year} {2007})}\BibitemShut {NoStop}%
\bibitem [{\citenamefont {Motsch}\ and\ \citenamefont {Tadmor}(2014)}]{MT14}%
  \BibitemOpen
  \bibfield  {author} {\bibinfo {author} {\bibfnamefont {S.}~\bibnamefont
  {Motsch}}\ and\ \bibinfo {author} {\bibfnamefont {E.}~\bibnamefont
  {Tadmor}},\ }\href@noop {} {\bibfield  {journal} {\bibinfo  {journal} {SIAM
  Review}\ }\textbf {\bibinfo {volume} {56}},\ \bibinfo {pages} {577} (\bibinfo
  {year} {2014})}\BibitemShut {NoStop}%
\bibitem [{\citenamefont {Chae}\ \emph {et~al.}(2015)\citenamefont {Chae},
  \citenamefont {Clouston}, \citenamefont {Hatzenbuehler}, \citenamefont
  {Kramer}, \citenamefont {Cooper}, \citenamefont {Wilson}, \citenamefont
  {Stephens-Davidowitz}, \citenamefont {Gold},\ and\ \citenamefont
  {Link}}]{CCHK15}%
  \BibitemOpen
  \bibfield  {author} {\bibinfo {author} {\bibfnamefont {D.~H.}\ \bibnamefont
  {Chae}}, \bibinfo {author} {\bibfnamefont {S.}~\bibnamefont {Clouston}},
  \bibinfo {author} {\bibfnamefont {M.~L.}\ \bibnamefont {Hatzenbuehler}},
  \bibinfo {author} {\bibfnamefont {M.~R.}\ \bibnamefont {Kramer}}, \bibinfo
  {author} {\bibfnamefont {H.~L.~F.}\ \bibnamefont {Cooper}}, \bibinfo {author}
  {\bibfnamefont {S.~M.}\ \bibnamefont {Wilson}}, \bibinfo {author}
  {\bibfnamefont {S.~I.}\ \bibnamefont {Stephens-Davidowitz}}, \bibinfo
  {author} {\bibfnamefont {R.~S.}\ \bibnamefont {Gold}}, \ and\ \bibinfo
  {author} {\bibfnamefont {B.~G.}\ \bibnamefont {Link}},\ }\href {\doibase
  10.1371/journal.pone.0122963} {\bibfield  {journal} {\bibinfo  {journal}
  {PLOS ONE}\ }\textbf {\bibinfo {volume} {10}},\ \bibinfo {pages} {1}
  (\bibinfo {year} {2015})}\BibitemShut {NoStop}%
\bibitem [{\citenamefont {Stephens-Davidowitz}(2017)}]{S17}%
  \BibitemOpen
  \bibfield  {author} {\bibinfo {author} {\bibfnamefont {S.}~\bibnamefont
  {Stephens-Davidowitz}},\ }\href@noop {} {\emph {\bibinfo {title} {Everybody
  Lies: Big Data, New Data, and What the Internet Can Tell Us About Who We
  Really Are}}}\ (\bibinfo  {publisher} {Dey Street Books},\ \bibinfo {year}
  {2017})\BibitemShut {NoStop}%
\bibitem [{\citenamefont {Stephens-Davidowitz}(2014)}]{S14}%
  \BibitemOpen
  \bibfield  {author} {\bibinfo {author} {\bibfnamefont {S.}~\bibnamefont
  {Stephens-Davidowitz}},\ }\href {\doibase
  http://dx.doi.org/10.1016/j.jpubeco.2014.04.010} {\bibfield  {journal}
  {\bibinfo  {journal} {Journal of Public Economics}\ }\textbf {\bibinfo
  {volume} {118}},\ \bibinfo {pages} {26 } (\bibinfo {year}
  {2014})}\BibitemShut {NoStop}%
\bibitem [{\citenamefont {Niu}(2013)}]{trends1}%
  \BibitemOpen
  \bibfield  {author} {\bibinfo {author} {\bibfnamefont {H.-J.}\ \bibnamefont
  {Niu}},\ }\href {\doibase 10.1111/jasp.12085} {\bibfield  {journal} {\bibinfo
   {journal} {Journal of Applied Social Psychology}\ }\textbf {\bibinfo
  {volume} {43}},\ \bibinfo {pages} {1228} (\bibinfo {year}
  {2013})}\BibitemShut {NoStop}%
\bibitem [{\citenamefont {Catalini}\ and\ \citenamefont
  {Tucker}(2016)}]{trends2}%
  \BibitemOpen
  \bibfield  {author} {\bibinfo {author} {\bibfnamefont {C.}~\bibnamefont
  {Catalini}}\ and\ \bibinfo {author} {\bibfnamefont {C.}~\bibnamefont
  {Tucker}},\ }\href {\doibase 10.3386/w22596} {\emph {\bibinfo {title}
  {Seeding the S-Curve? The Role of Early Adopters in Diffusion}}},\ \bibinfo
  {type} {Working Paper}\ \bibinfo {number} {22596}\ (\bibinfo  {institution}
  {National Bureau of Economic Research},\ \bibinfo {year} {2016})\BibitemShut
  {NoStop}%
\bibitem [{\citenamefont {Bapna}\ and\ \citenamefont
  {Umyarov}(2015)}]{purchasing}%
  \BibitemOpen
  \bibfield  {author} {\bibinfo {author} {\bibfnamefont {R.}~\bibnamefont
  {Bapna}}\ and\ \bibinfo {author} {\bibfnamefont {A.}~\bibnamefont
  {Umyarov}},\ }\href {\doibase 10.1287/mnsc.2014.2081} {\bibfield  {journal}
  {\bibinfo  {journal} {Management Science}\ }\textbf {\bibinfo {volume}
  {61}},\ \bibinfo {pages} {1902} (\bibinfo {year} {2015})},\ \Eprint
  {http://arxiv.org/abs/http://dx.doi.org/10.1287/mnsc.2014.2081}
  {http://dx.doi.org/10.1287/mnsc.2014.2081} \BibitemShut {NoStop}%
\bibitem [{\citenamefont {Chen}(2012)}]{norms1}%
  \BibitemOpen
  \bibfield  {author} {\bibinfo {author} {\bibfnamefont {X.}~\bibnamefont
  {Chen}},\ }\href {\doibase 10.1111/j.1750-8606.2011.00187.x} {\bibfield
  {journal} {\bibinfo  {journal} {Child Development Perspectives}\ }\textbf
  {\bibinfo {volume} {6}},\ \bibinfo {pages} {27} (\bibinfo {year}
  {2012})}\BibitemShut {NoStop}%
\bibitem [{\citenamefont {Rajtmajer}\ \emph {et~al.}(2016)\citenamefont
  {Rajtmajer}, \citenamefont {Squicciarini}, \citenamefont {Griffin},
  \citenamefont {Karumanchi},\ and\ \citenamefont {Tyagi}}]{RSGK16}%
  \BibitemOpen
  \bibfield  {author} {\bibinfo {author} {\bibfnamefont {S.}~\bibnamefont
  {Rajtmajer}}, \bibinfo {author} {\bibfnamefont {A.}~\bibnamefont
  {Squicciarini}}, \bibinfo {author} {\bibfnamefont {C.}~\bibnamefont
  {Griffin}}, \bibinfo {author} {\bibfnamefont {S.}~\bibnamefont {Karumanchi}},
  \ and\ \bibinfo {author} {\bibfnamefont {A.}~\bibnamefont {Tyagi}},\ }in\
  \href {http://dl.acm.org/citation.cfm?id=2936924.2937025} {\emph {\bibinfo
  {booktitle} {Proceedings of the 2016 International Conference on Autonomous
  Agents \&\#38; Multiagent Systems}}}\ (\bibinfo {year} {2016})\ pp.\ \bibinfo
  {pages} {680--688}\BibitemShut {NoStop}%
\bibitem [{\citenamefont {Hong}\ and\ \citenamefont
  {Espelage}(2012)}]{bullying1}%
  \BibitemOpen
  \bibfield  {author} {\bibinfo {author} {\bibfnamefont {J.~S.}\ \bibnamefont
  {Hong}}\ and\ \bibinfo {author} {\bibfnamefont {D.~L.}\ \bibnamefont
  {Espelage}},\ }\href {\doibase http://dx.doi.org/10.1016/j.avb.2012.03.003}
  {\bibfield  {journal} {\bibinfo  {journal} {Aggression and Violent Behavior}\
  }\textbf {\bibinfo {volume} {17}},\ \bibinfo {pages} {311 } (\bibinfo {year}
  {2012})}\BibitemShut {NoStop}%
\bibitem [{\citenamefont {Christakis}\ and\ \citenamefont
  {Fowler}(2007)}]{obesity}%
  \BibitemOpen
  \bibfield  {author} {\bibinfo {author} {\bibfnamefont {N.~A.}\ \bibnamefont
  {Christakis}}\ and\ \bibinfo {author} {\bibfnamefont {J.~H.}\ \bibnamefont
  {Fowler}},\ }\href {\doibase 10.1056/NEJMsa066082} {\bibfield  {journal}
  {\bibinfo  {journal} {New England Journal of Medicine}\ }\textbf {\bibinfo
  {volume} {357}},\ \bibinfo {pages} {370} (\bibinfo {year} {2007})},\ \bibinfo
  {note} {pMID: 17652652},\ \Eprint
  {http://arxiv.org/abs/http://dx.doi.org/10.1056/NEJMsa066082}
  {http://dx.doi.org/10.1056/NEJMsa066082} \BibitemShut {NoStop}%
\bibitem [{\citenamefont {Renna}\ \emph {et~al.}(2008)\citenamefont {Renna},
  \citenamefont {Grafova},\ and\ \citenamefont {Thakur}}]{weight}%
  \BibitemOpen
  \bibfield  {author} {\bibinfo {author} {\bibfnamefont {F.}~\bibnamefont
  {Renna}}, \bibinfo {author} {\bibfnamefont {I.~B.}\ \bibnamefont {Grafova}},
  \ and\ \bibinfo {author} {\bibfnamefont {N.}~\bibnamefont {Thakur}},\
  }\href@noop {} {\bibfield  {journal} {\bibinfo  {journal} {Economics \& Human
  Biology}\ }\textbf {\bibinfo {volume} {6}},\ \bibinfo {pages} {377 }
  (\bibinfo {year} {2008})},\ \bibinfo {note} {symposium on the Economics of
  Obesity}\BibitemShut {NoStop}%
\bibitem [{\citenamefont {Mednick}\ \emph {et~al.}(2010)\citenamefont
  {Mednick}, \citenamefont {Christakis},\ and\ \citenamefont {Fowler}}]{sleep}%
  \BibitemOpen
  \bibfield  {author} {\bibinfo {author} {\bibfnamefont {S.~C.}\ \bibnamefont
  {Mednick}}, \bibinfo {author} {\bibfnamefont {N.~A.}\ \bibnamefont
  {Christakis}}, \ and\ \bibinfo {author} {\bibfnamefont {H.~J.}\ \bibnamefont
  {Fowler}},\ }\href@noop {} {\bibfield  {journal} {\bibinfo  {journal} {PLoS
  ONE}\ }\textbf {\bibinfo {volume} {5}} (\bibinfo {year} {2010})}\BibitemShut
  {NoStop}%
\bibitem [{\citenamefont {Harakeh}\ and\ \citenamefont
  {Vollebergh}(2012)}]{smoking}%
  \BibitemOpen
  \bibfield  {author} {\bibinfo {author} {\bibfnamefont {Z.}~\bibnamefont
  {Harakeh}}\ and\ \bibinfo {author} {\bibfnamefont {W.~A.}\ \bibnamefont
  {Vollebergh}},\ }\href@noop {} {\bibfield  {journal} {\bibinfo  {journal}
  {Drug and Alcohol Dependence}\ }\textbf {\bibinfo {volume} {121}},\ \bibinfo
  {pages} {220 } (\bibinfo {year} {2012})}\BibitemShut {NoStop}%
\bibitem [{\citenamefont {Friedkin}(2006)}]{friedkin2006}%
  \BibitemOpen
  \bibfield  {author} {\bibinfo {author} {\bibfnamefont {N.~E.}\ \bibnamefont
  {Friedkin}},\ }\href@noop {} {\emph {\bibinfo {title} {A structural theory of
  social influence}}},\ Vol.~\bibinfo {volume} {13}\ (\bibinfo  {publisher}
  {Cambridge University Press},\ \bibinfo {year} {2006})\BibitemShut {NoStop}%
\bibitem [{\citenamefont {Brewer}(1991)}]{B91}%
  \BibitemOpen
  \bibfield  {author} {\bibinfo {author} {\bibfnamefont {M.~B.}\ \bibnamefont
  {Brewer}},\ }\href@noop {} {\bibfield  {journal} {\bibinfo  {journal}
  {Personality and Social Psychology Bulletin}\ }\textbf {\bibinfo {volume}
  {17}},\ \bibinfo {pages} {475} (\bibinfo {year} {1991})}\BibitemShut
  {NoStop}%
\bibitem [{\citenamefont {van Baaren}\ \emph {et~al.}(2009)\citenamefont {van
  Baaren}, \citenamefont {Janssen}, \citenamefont {Chartrand},\ and\
  \citenamefont {Dijksterhuis}}]{BJCD09}%
  \BibitemOpen
  \bibfield  {author} {\bibinfo {author} {\bibfnamefont {R.}~\bibnamefont {van
  Baaren}}, \bibinfo {author} {\bibfnamefont {L.}~\bibnamefont {Janssen}},
  \bibinfo {author} {\bibfnamefont {T.~L.}\ \bibnamefont {Chartrand}}, \ and\
  \bibinfo {author} {\bibfnamefont {A.}~\bibnamefont {Dijksterhuis}},\
  }\href@noop {} {\bibfield  {journal} {\bibinfo  {journal} {Philosophical
  Trans. of the Royal Society B}\ }\textbf {\bibinfo {volume} {364}},\ \bibinfo
  {pages} {2381} (\bibinfo {year} {2009})}\BibitemShut {NoStop}%
\bibitem [{\citenamefont {Gill}(1991)}]{G91}%
  \BibitemOpen
  \bibfield  {author} {\bibinfo {author} {\bibfnamefont {J.}~\bibnamefont
  {Gill}},\ }\href@noop {} {\bibfield  {journal} {\bibinfo  {journal} {Applied
  Numerical Mathematics}\ }\textbf {\bibinfo {volume} {8}},\ \bibinfo {pages}
  {469 } (\bibinfo {year} {1991})}\BibitemShut {NoStop}%
\bibitem [{\citenamefont {Godsil}\ and\ \citenamefont {Royle}(2001)}]{GR01}%
  \BibitemOpen
  \bibfield  {author} {\bibinfo {author} {\bibfnamefont {C.}~\bibnamefont
  {Godsil}}\ and\ \bibinfo {author} {\bibfnamefont {G.}~\bibnamefont {Royle}},\
  }\href@noop {} {\emph {\bibinfo {title} {Algebraic Graph Theory}}}\ (\bibinfo
   {publisher} {Springer},\ \bibinfo {year} {2001})\BibitemShut {NoStop}%
\bibitem [{\citenamefont {Royden}\ and\ \citenamefont
  {Fitzpatrick}(2009)}]{RF09}%
  \BibitemOpen
  \bibfield  {author} {\bibinfo {author} {\bibfnamefont {H.~L.}\ \bibnamefont
  {Royden}}\ and\ \bibinfo {author} {\bibfnamefont {P.~M.}\ \bibnamefont
  {Fitzpatrick}},\ }\href@noop {} {\emph {\bibinfo {title} {{Real
  Analysis}}}},\ \bibinfo {edition} {4th}\ ed.\ (\bibinfo  {publisher}
  {Pearson},\ \bibinfo {year} {2009})\BibitemShut {NoStop}%
\bibitem [{\citenamefont {Lorentzen}(1990)}]{Lorentzen90}%
  \BibitemOpen
  \bibfield  {author} {\bibinfo {author} {\bibfnamefont {L.}~\bibnamefont
  {Lorentzen}},\ }\href@noop {} {\bibfield  {journal} {\bibinfo  {journal} {J.
  Computational and Applied Mathematics}\ }\textbf {\bibinfo {volume} {32}},\
  \bibinfo {pages} {169} (\bibinfo {year} {1990})}\BibitemShut {NoStop}%
\bibitem [{\citenamefont {Bertsekas}(1999)}]{Bert99}%
  \BibitemOpen
  \bibfield  {author} {\bibinfo {author} {\bibfnamefont {D.~P.}\ \bibnamefont
  {Bertsekas}},\ }\href@noop {} {\emph {\bibinfo {title} {{Nonlinear
  Programming}}}},\ \bibinfo {edition} {2nd}\ ed.\ (\bibinfo  {publisher}
  {Athena Scientific},\ \bibinfo {year} {1999})\BibitemShut {NoStop}%
\bibitem [{\citenamefont {Barab\'{a}si}\ and\ \citenamefont
  {Albert}(1999)}]{BA99}%
  \BibitemOpen
  \bibfield  {author} {\bibinfo {author} {\bibfnamefont {A.}~\bibnamefont
  {Barab\'{a}si}}\ and\ \bibinfo {author} {\bibfnamefont {R.}~\bibnamefont
  {Albert}},\ }\href@noop {} {\bibfield  {journal} {\bibinfo  {journal}
  {Science}\ }\textbf {\bibinfo {volume} {286}},\ \bibinfo {pages} {509}
  (\bibinfo {year} {1999})}\BibitemShut {NoStop}%
\bibitem [{\citenamefont {Roughgarden}(2003)}]{Rough03}%
  \BibitemOpen
  \bibfield  {author} {\bibinfo {author} {\bibfnamefont {T.}~\bibnamefont
  {Roughgarden}},\ }\href@noop {} {\bibfield  {journal} {\bibinfo  {journal}
  {J. Computer and System Sciences}\ }\textbf {\bibinfo {volume} {67}},\
  \bibinfo {pages} {341} (\bibinfo {year} {2003})}\BibitemShut {NoStop}%
\bibitem [{\citenamefont {Madan}\ \emph {et~al.}(2012)\citenamefont {Madan},
  \citenamefont {Cebrian}, \citenamefont {Moturu}, \citenamefont {Farrahi},\
  and\ \citenamefont {Pentland}}]{madan2012}%
  \BibitemOpen
  \bibfield  {author} {\bibinfo {author} {\bibfnamefont {A.}~\bibnamefont
  {Madan}}, \bibinfo {author} {\bibfnamefont {M.}~\bibnamefont {Cebrian}},
  \bibinfo {author} {\bibfnamefont {S.}~\bibnamefont {Moturu}}, \bibinfo
  {author} {\bibfnamefont {K.}~\bibnamefont {Farrahi}}, \ and\ \bibinfo
  {author} {\bibfnamefont {A.}~\bibnamefont {Pentland}},\ }\href@noop {}
  {\bibfield  {journal} {\bibinfo  {journal} {IEEE Pervasive Computing}\
  }\textbf {\bibinfo {volume} {11}},\ \bibinfo {pages} {36} (\bibinfo {year}
  {2012})}\BibitemShut {NoStop}%
\bibitem [{\citenamefont {Griffin}\ \emph {et~al.}(2016)\citenamefont
  {Griffin}, \citenamefont {Rajtamajer}, \citenamefont {Squicciarini},\ and\
  \citenamefont {Belmonte}}]{CGT16}%
  \BibitemOpen
  \bibfield  {author} {\bibinfo {author} {\bibfnamefont {C.}~\bibnamefont
  {Griffin}}, \bibinfo {author} {\bibfnamefont {S.}~\bibnamefont {Rajtamajer}},
  \bibinfo {author} {\bibfnamefont {A.}~\bibnamefont {Squicciarini}}, \ and\
  \bibinfo {author} {\bibfnamefont {A.}~\bibnamefont {Belmonte}},\ }\href@noop
  {} {\bibfield  {journal} {\bibinfo  {journal} {Submitted to SIAM J. Applied
  Dynamical Systems}\ } (\bibinfo {year} {2016})}\BibitemShut {NoStop}%
\end{thebibliography}%

\end{document}